\documentclass[pra,reprint,superscriptaddress,amsmath,amssymb,aps]{revtex4-2}
\usepackage{float}
\makeatletter
\let\newfloat\newfloat@ltx
\makeatother

\makeatletter
\renewcommand\@make@capt@title[2]{%
  {\textbf{#1}}\@caption@fignum@sep #2\quad}
\def\@caption@font{\normalsize}
\makeatother

\usepackage{makecell}

\usepackage{algorithm}
\usepackage{algpseudocode}
\makeatletter
\renewcommand{\fnum@algorithm}{\fname@algorithm~\thealgorithm}
\makeatother

\newcounter{panel}[figure]
\renewcommand{\thepanel}{\alph{panel}}
\newcommand{\panel}[1]{%
  \refstepcounter{panel}%
  \label{#1}%
  \par\vspace{2pt}(\thepanel)}

\usepackage{multirow}
\usepackage{array}
\usepackage{booktabs}
\usepackage{amsmath,amssymb,amsthm,mathtools,mathrsfs}

\newtheorem{theorem}{Theorem}

\newtheorem{problem}{Problem}

\newtheorem{apptheorem}{Theorem}[section]
\newtheorem{applemma}{Lemma}[subsection]

\usepackage[dvipsnames,table]{xcolor}
\usepackage{booktabs,multirow}
\definecolor{DarkGreen}{rgb}{0.0,0.3,0.0}
\definecolor{faintline}{gray}{0.85}
\usepackage{graphicx}
\usepackage{tikz}
\usetikzlibrary{decorations.pathreplacing,calc}
\usepackage{quantikz}

\usepackage{xcolor}

\usepackage{bm}

\usepackage[colorlinks=true,
            linkcolor=red,
            citecolor=blue,
            urlcolor=red]{hyperref}
\usepackage{cleveref}

\crefname{definition}{Def.}{Defs.}
\Crefname{definition}{Definition}{Definitions}
\crefname{problem}{Prob.}{Probs.}
\Crefname{problem}{Problem}{Problems}
\crefname{lemma}{Lemma}{Lemmas}
\Crefname{lemma}{Lemma}{Lemmas}
\crefname{equation}{Eqn.}{Eqns.}

\makeatletter
\newcommand{\problabel}[2]{%
  \def\@currentlabel{#1}%
  \label{#2}%
}
\makeatother

\newcommand{\orcid}[1]{\href{https://orcid.org/#1}{\includegraphics[height = 2ex]{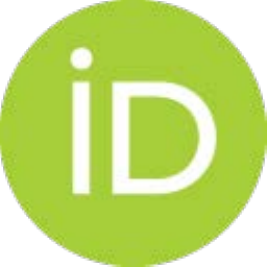}}}

\usepackage{float}
\begin{document}

\title{Unitary Channel Testing Under a Depolarizing Noise Assumption}

\author{Hirak Ghosh \orcid{0009-0007-9276-1542}}
\email{Hirak.Ghosh@warwick.ac.uk}
\affiliation{Department of Physics, University of Warwick, Coventry, CV4 7AL, UK}

\author{Andrew Jackson \orcid{0000-0002-5981-1604}}
\affiliation{Department of Physics, University of Warwick, Coventry, CV4 7AL, UK}

\author{Animesh Datta \orcid{0000-0003-4021-4655}}
\affiliation{Department of Physics, University of Warwick, Coventry, CV4 7AL, UK}

\date{\today}

\begin{abstract}
We present fast algorithms -- under the depolarizing noise assumption, often made in fault-tolerant quantum computations -- to test its strength.
Our optimal algorithms answer the following question: is the quantum channel implemented by a given black box identical to a target unitary or $\varepsilon$-far from it in the diamond distance, assuming that the deviation is a depolarizing channel with unknown parameter?
Our algorithm has a query complexity of $\Theta(1/\varepsilon).$
The query complexity of the relaxed problem of testing whether the black-box channel is $\varepsilon_1$-close to a target unitary or $\varepsilon_2$-far in the diamond distance is $\Theta\bigl(\varepsilon_2/(\varepsilon_2 - \varepsilon_1)^2\bigr).$
In both cases, we provide matching lower bounds that hold even for adaptive, ancilla-assisted protocols with multi-outcome incoherent measurements.
\end{abstract}

\maketitle

Performing tomography or learning a quantum process in the diamond distance typically has exponential costs in the number of qubits \cite{surawy2022projected, oufkir2023sample, haah2023query, chen2024tight, mele2025optimal}.
This presents a fundamental barrier to improving quantum computing hardware towards the fault-tolerant regime whose thresholds are quantified in the diamond distance~\cite{doi:10.1137/S0097539799359385}.
Furthermore, every threshold is calculated for an assumed noise model whose strength must be established in the hardware.

Benchmarking methods are typically deployed to
extract average fidelity metrics for charting progress towards fault tolerance without learning the underlying quantum channels~\cite{Malhotra2024, Proctor2025Benchmarking, rohe2025quantumcomputerbenchmarkingexplorative}. In practice, cycle benchmarking and its variants are often used~\cite{Erhard2019CycleBenchmarking, HarperFlammiaWallman2020, ChenLearnability2023}.
However, the sample (or query) complexity of cycle benchmarking scales exponentially in the qubit count even for structured classes, such as Pauli channels~\cite{FlammiaWallman2020,FawziOufkirFranca2025}. 
Equally importantly, the outputs of benchmarking, such as average gate fidelity correspond, in general, poorly with the diamond distance~\cite{Sanders_2016}.
The inability to measure the diamond distance without exponential complexity implies that benchmarking fails to surmount the
fundamental barrier noted above.  

This motivates the seeking schemes to test the quality of quantum computing hardware directly in the diamond norm.
Emanating from the field of property testing~\cite{gs007}, determining the proximity of quantum channels (often called certification) in the diamond distance also appears to be afflicted with exponential complexities~\cite{fawzi2023quantum,rosenthal2024quantum,chen2026strict} in many cases. 
To identify the roots of these complexities, recent studies have explored different distance measures such as the average-case imitation diamond (ACID) norm~\cite{rosenthal2024quantum} or additional resources such as a quantum memory~\cite{chen2025efficient}. 
However, exponential complexities seems to persist even when testing unitary channels assuming coherent noise~\cite{jeon2025query}, albeit with a lower cost.

In calculating quantum fault-tolerant thresholds, the  circuit-level noise model most often assumed for quantum hardware is a depolarizing channel~\cite{Raussendorf_2007,PRXQuantum.5.030352,sblg-fbq4}.
However, the query complexity of testing the strength of this assumption in actual hardware is unknown. We dub this Problem~\ref{prob:pre-reduction}.

In this Letter, we show that the query complexity of solving Problem~\ref{prob:pre-reduction} is independent of the number of qubits involved. This makes our algorithms fast.
More precisely, we show that there exists a non-adaptive, ancilla-free tester with optimal query complexity $\Theta(1/\varepsilon)$ for distinguishing a perfect unitary channel implementation from one that is $\varepsilon$-far from the target in the diamond distance. We further show that this optimal scaling persists even when we restrict the number of inputs and individual measurements to $\mathcal{O}(1)$, by establishing a tradeoff between circuit depth and the number of individual state-measurements pairs used. We also give a tolerant tester that, under the same depolarizing noise assumption, distinguishes implementations that are $\varepsilon_1$-close to the target from those that are $\varepsilon_2$-far in diamond distance, with optimal query complexity $\Theta\bigl(\varepsilon_2/(\varepsilon_2 - \varepsilon_1)^2\bigr)$. 
A summary of our results is presented in Table~\ref{table 1}. 
A comparison with related known results is in Table~\ref{Table 2}. 

\begin{table}[t]
\centering
\setlength{\tabcolsep}{1pt}
\renewcommand{\arraystretch}{1.2}
\begin{tabular}{|c|c|c|c|c|}
\hline
\makecell{\textbf{Ancilla}} & \makecell{\textbf{Adapt-}\\\textbf{ivity}} & 
\makecell{\textbf{Measure-}\\\textbf{ment}} & \makecell{\textbf{Standard}\\\textbf{Testing}} & \makecell{\textbf{Tolerant}\\\textbf{Testing}} \\
\noalign{\hrule height 1pt}
No  & No  & Binary &
$\Theta(1/\varepsilon)$ &
$\Theta\left(\varepsilon_2/(\varepsilon_2 - \varepsilon_1)^2\right)$ \\
&&&Th.~\ref{TH : Standard tester, optimality},~\ref{TH:m-standard-tradeoff}&Th.~\ref{TH : Tol tester, optimality}\\\hline
Yes & Yes & Multi &
$\Theta(1/\varepsilon)$ &
$\Theta\left(\varepsilon_2/(\varepsilon_2 - \varepsilon_1)^2\right)$ \\&&&Th.~\ref{TH : Standard tester, optimality}&Th.~\ref{TH : Tol tester, optimality}\\
\hline
\end{tabular}
\caption{\normalsize Query complexity of diamond-distance testing under different access models. Here, $\Theta$ denotes optimal query complexity, that is, the lower and upper bounds match up to constant factors. For clarity, the $\Theta(\cdot)$ notation hides the dependence on the error parameter $\delta$.}
\label{table 1}
\end{table}

We avoid entangled inputs, collective measurements, and quantum memories as inspired by practical testing scenarios.
However, we do assume perfect state preparation and measurement. 
The invariance of the diamond distance under unitary conjugation reduces our Problem~\ref{prob:pre-reduction} to 
Problem~\ref{prob:testing-depolarizing} -- that of testing whether the depolarizing channel is either close or far from the identity channel. The learning version of this single-parameter problem is estimation upto an additive error~\cite{sasaki2002optimal,ji2008parameter}.  A Hoeffding's inequality shows that a non-adaptive, ancilla-free learner with two outcome measurements yields $\mathcal{O}(1/\varepsilon^2)$ query complexity (Lemma~\ref{eq:depol-lemma}). 
We thus give the corresponding optimal \emph{testing} bounds and establish that resources like adaptivity, ancilla, and multi-outcome POVMs do not improve query complexity.

\begin{table}[h]
\centering
\setlength{\tabcolsep}{5pt}
\renewcommand{\arraystretch}{1.25}
\begin{tabular}{|cl|c|}
\hline
\makecell{} & \makecell{Task} & \makecell{\textbf{Complexity}} \\
\noalign{\hrule height 1pt}
1. & Learning general channels & $\Theta(d^4/\varepsilon^2)$ \cite{mele2025optimal} \\
\hline
2. & Learning unitary channels & $\Theta(d^2/\varepsilon)$ \cite{haah2023query} \\
\hline
3. & Testing general channels & $\Omega(\sqrt d/\varepsilon)$ \cite{rosenthal2024quantum} \\
\hline
4. & Testing unitary channels &  $\Theta( d/\varepsilon)$ \cite{chen2026strict} \\
\hline
5. & \makecell[l]{Testing unitary channels under \\ coherent noise assumption}  & $\Theta(\sqrt d/\varepsilon)$ \cite{jeon2025query} \\\hline
6. & \makecell[l]{Testing unitary channels under \\ depolarizing noise assumption}  & \makecell[c]{$\Theta(1/\varepsilon)$ \\ This Letter}\\
\hline
\end{tabular}
\caption{\normalsize Complexities of various learning and testing tasks, using an access model permitting ancilla, coherence, and adaptivity. All tasks quantifying proximity are in the diamond distance. 
Our algorithms are fast as their time complexities are independent of the number of qubits.
$\Omega(\cdot)$ denotes a lower bound while $\Theta (\cdot)$ denotes optimality with matching lower and upper bounds.
For $n$ qubits, $d = 2^n.$
}
\label{Table 2}
\end{table}

\emph{Preliminaries}:
Let $\mathcal{H} \cong \mathbb{C}^d$ be a $d$-dimensional Hilbert space, where $d \in \mathbb{N}$.
Quantum states are density operators $\rho \in \mathcal{L}(\mathcal{H})$ with $\rho \ge 0,$ $\mathrm{Tr}(\rho)=1$.
A quantum channel is a completely positive trace-preserving (CPTP) linear map $\mathcal{N} : \mathcal{L}(\mathcal{H}) \rightarrow \mathcal{L}(\mathcal{H}).$
A unitary channel is defined as \(\mathcal{N}_U(\rho) = U \rho U^\dagger\), where \(U \in \mathcal{U}(d)\) is a unitary operator on a \(d\)-dimensional Hilbert space.
A depolarizing channel with a parameter (strength) $\lambda \in [0,1]$ $\mathcal{N}_{\lambda}$, is a quantum channel such that for all $\rho \in \mathcal{L}(\mathcal{H})$,
\begin{align}
\label{eq:dep}
    \mathcal{N}_{\lambda}(\rho) = (1-\lambda) \rho + \lambda\dfrac{\mathrm{Tr}(\rho)}{d}\,\mathbb{I}_d.
\end{align}
    A positive operator-valued measure (POVM) is a model of measurements defined by a set of operators, $\mathbb{M} = \{M_x\}_x$, satisfying $M_x \ge 0$ and $\sum_x M_x = \mathbb{I}$.
 For an input state, $\rho$, acted upon by a channel $\mathcal{N}$, the probability of obtaining outcome $x$ is $\Pr(x)
= \operatorname{Tr}\left(M_x\, \mathcal{N}(\rho)\right).$
The diamond distance $d_{\diamond}(\cdot , \cdot)$ between two channels $\mathcal{M}$ and $\mathcal{N}$ is 
\begin{align}
    d_{\diamond}(\mathcal{M}, \mathcal{N}) = \sup_{\rho \in \mathcal{L}(\mathcal{H}_a \otimes \mathcal{H})}  \|(\mathrm{id}_{\mathcal{H}_{a}} \otimes (\mathcal{M}-\mathcal{N})(\rho))\|_{1},
\end{align}
with $\mathcal{H}_a \cong \mathbb{C}^{d_a}$ being an auxiliary system. Owing to Schmidt decomposition, we can always set $d=d_a$.

\emph{The problem}:
We are given black-box access to \(N \in \mathbb{N}\) uses of a channel 
\(
\mathcal{N}_{U_*} = \mathcal{N}_U \circ \mathcal{N}_\lambda,
\) 
where $\mathcal{N}_U(\rho)$
is a \emph{known} unitary channel and \(\mathcal{N}_\lambda\) is a depolarizing channel with an \emph{unknown} depolarizing parameter $\lambda$.
The order of application of $\mathcal{N}_U$ and $\mathcal{N}_{\lambda}$ does not matter here, since depolarizing channel is unitarily covariant. 
See Appendix \ref{sec:AppA}, Lemma~\ref{Dep_channel covariance}.
Our testing task is as in Problem~\ref{prob:pre-reduction}.

\begin{problem}
\label{prob:pre-reduction}
    Distinguish, with success probability at least \(1-\delta\), between $H_0: d_\diamond(\mathcal{N}_{U_*},\mathcal{N}_U) \le \varepsilon_1$ and $H_1: d_\diamond(\mathcal{N}_{U_*},\mathcal{N}_U) \ge \varepsilon_2$ for fixed parameters \(0 < \varepsilon_1 < \varepsilon_2 < 1\) and \(\delta \in (0,1/2)\).
\end{problem}

We provide \emph{lower and upper bounds} on the query complexity (or the necessary and sufficient number of channel uses) \(N\)  required to 
solve Problem~\ref{prob:pre-reduction}.
We do so both in the standard  \((\varepsilon_1 = 0)\) and the tolerant case \((\varepsilon_1 > 0)\). 

All our bounds are attained in the following manner. The input state
$\rho \in \mathcal{L}(\mathcal H)$ is first passed through the implemented
channel $\mathcal{N}_{U_*}$ and then measured with the POVM
$\mathbb{M}_{U}=\{U M_x U^{\dagger}\}_x$.
Owing to the invariance of the diamond
distance under unitary conjugation,
\(d_{\diamond}(\mathcal{N}_{U_*},\mathcal{N}_{U})
=
d_{\diamond}(\mathcal{N}_{\lambda},\operatorname{id}),\)
this implementation is equivalent to the procedure of preparing $\rho$,
applying the channel $\mathcal{N}_{\lambda}$, and measuring the resulting state
$\mathcal{N}_{\lambda}(\rho)$ with the POVM $\mathbb M = \{M_x\}_x$. The corresponding
outcome statistics are
\(
\Pr(x)=\operatorname{Tr}\!\left(\mathcal{N}_{\lambda}(\rho)M_x\right).
\) 

Thus, Problem~\ref{prob:pre-reduction} reduces to Problem~\ref{prob:testing-depolarizing}.

 \begin{problem}[Depolarizing Channel Testing]
 \label{prob:testing-depolarizing}
 Distinguish, with success probability at least \(1-\delta\), between 
\begin{enumerate}
    \item[\emph{\textbf{P2.1}}] 
    \problabel{P2.1}{prob:P2.1}%
    \textit{Standard Testing:}  
    $H_0:\mathcal{N}_\lambda = \operatorname{id}$  and $H_1:d_\diamond(\mathcal{N}_\lambda,\operatorname{id}) \ge \varepsilon$,
    where $\varepsilon \in (0,1)$.
    \item[\emph{\textbf{P2.2}}] 
    \problabel{P2.2}{prob:P2.2}%
    \textit{Tolerant Testing:} 
    $H_0:d_\diamond(\mathcal{N}_\lambda,\operatorname{id}) \le \varepsilon_1$
    and $H_1:d_\diamond(\mathcal{N}_\lambda,\operatorname{id}) \ge \varepsilon_2$,
    where $0 < \varepsilon_1 < \varepsilon_2 < 1$.
\end{enumerate}
\end{problem}
If \(m \in \mathbb{N}\) sequential applications of the black box are allowed, then the appropriate measurement that reduces Problem~\ref{prob:pre-reduction} to Problem~\ref{prob:testing-depolarizing} is
\(\widetilde{\mathbb{M}}_U
=
\left\{
U^{m} M_x \left(U^\dagger\right)^{m}
\right\}_x .\)
We note that Problems~\ref{prob:pre-reduction} and~\ref{prob:testing-depolarizing} have same the query complexity but different computational complexities.

\begin{figure*}[!t]
\centering

\begin{minipage}[t]{0.31\textwidth}
  \centering
  \includegraphics[page=4,width=\linewidth]{No_Ancilla_No_Adaptivity_Binary_Measurements.pdf}
  \panel{fig:combined:a}
\end{minipage}
\hfill
{\color{faintline}\vrule width 0.4pt}
\hfill
\begin{minipage}[t]{0.31\textwidth}
  \centering
  \includegraphics[page=1,width=\linewidth]{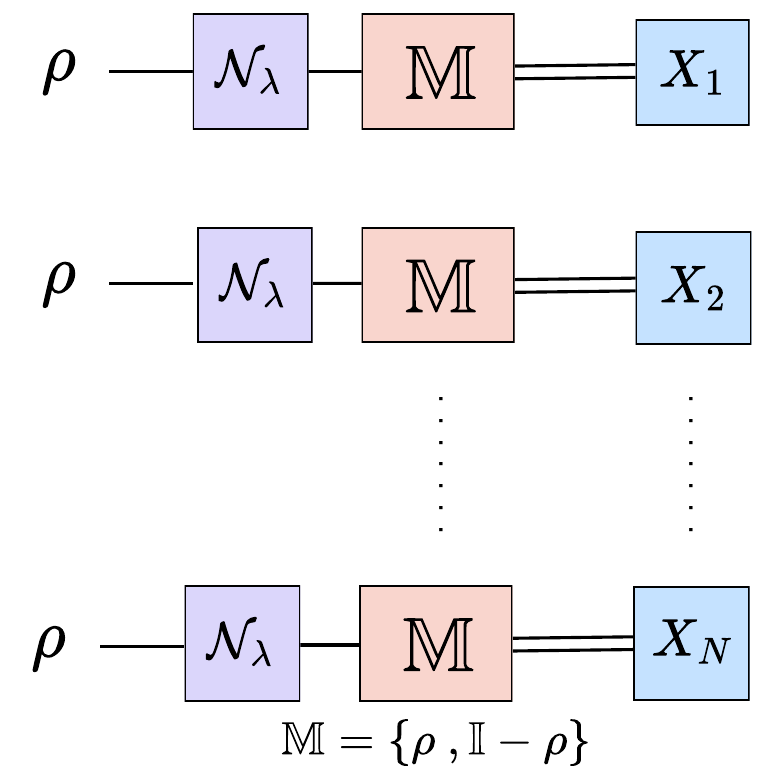}
  \panel{fig:combined:b}
\end{minipage}
\hfill
{\color{faintline}\vrule width 0.4pt}
\hfill
\begin{minipage}[t]{0.31\textwidth}
  \centering
  \includegraphics[page=1,width=\linewidth]{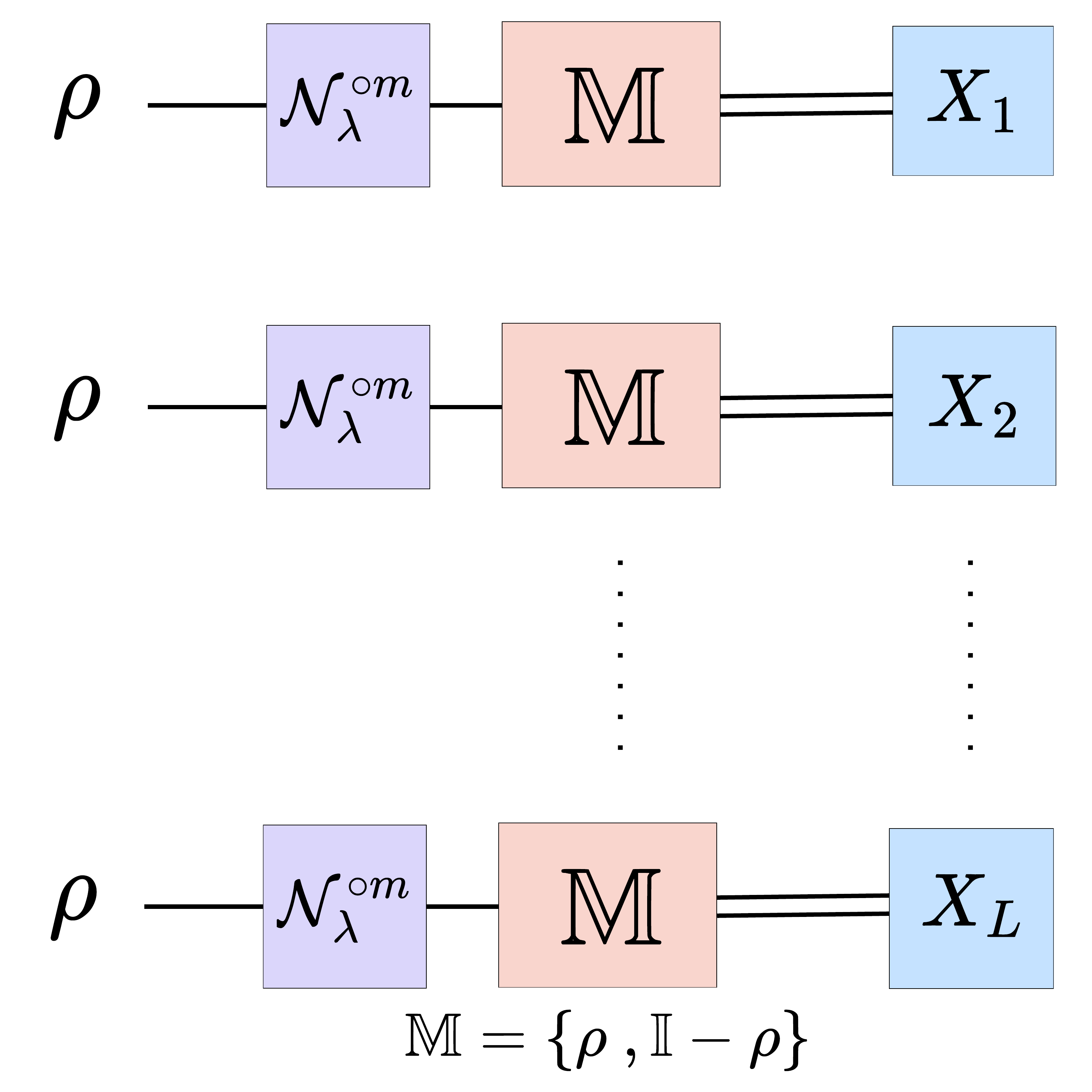}
  \panel{fig:combined:c}
\end{minipage}

\caption{\normalsize
(a) Ancilla-assisted, adaptive tester with arbitrary (finite) sized POVMs. Green arrows denote classical feed-forward: state preparations and measurements may depend on some or all prior measurement outcomes.
(b) Non-adaptive, ancilla-free tester with fixed inputs and a two-outcome measurement.
(c) \(m\)-sequential non-adaptive tester with fixed inputs and a two-outcome POVM.}
\label{fig:combined}
\end{figure*}

\emph{Access Models}:
Before presenting the query complexity for solving Problem~\ref{prob:testing-depolarizing}, 
we summarize the resources relevant to our complexity bounds. These are called access models, and illustrated in Fig.~\ref{fig:combined}.

In the most general setting with incoherent measurements (Fig.~\ref{fig:combined}\ref{fig:combined:a}), the tester is granted the following resources in sequence. 
First, the tester may attach an auxiliary system $\mathcal{H}' \cong \mathbb{C}^{d'}$ of arbitrary finite dimension, fixing the model: in the ancilla-assisted model the input space is $\mathcal{H}' \otimes \mathcal{H}$ and the channel acts as $\mathrm{id}_{\mathcal{H}'} \otimes \mathcal{N}$, while in the ancilla-free model the input space is $\mathcal{H}$ alone. 
Second, at each trial $t \in \mathbb{N}$, the tester chooses an input state $\rho^{(t)}$ on the input space. Third, after each channel use, the tester performs an individual measurement $\mathbb{M}^{(t)} = \{M_i^{(t)}\}_i$ on the corresponding output system; here, ``individual'' means that the measurement acts only on the output of a single channel use, with no collective measurement across different trials. The resulting single-shot outcome distribution is
\[
p_i^{(t)}
=
\operatorname{Tr}\!\left[\bigl((\mathrm{id}_{\mathcal{H}'} \otimes \mathcal{N})(\rho^{(t)})\bigr) M_i^{(t)}\right],
\]
which reduces to $p_i^{(t)} = \operatorname{Tr}[\mathcal{N}(\rho^{(t)}) M_i^{(t)}]$ in the ancilla-free case. Finally, the tester may be adaptive: the choices of $\rho^{(t)}$ and $\mathbb{M}^{(t)}$ may depend on the full history of previously observed outcomes. We refer to the resulting model, which allows ancillas,
adaptivity, and arbitrary individual measurements, as the \emph{full-access
model}; see Fig.~\ref{fig:combined}\ref{fig:combined:a}.

We establish lower bounds for Problems~\ref{prob:P2.1} and~\ref{prob:P2.2} in the \emph{full-access model} in Appendix~\ref{app:LowerBound}, specifically in Theorems~\ref{Th:LB_Tol} and~\ref{Th:LB_Std}, respectively.

For the upper bounds, we work in ancilla-free, non-adaptive access models with binary POVMs. Let \(m \in \mathbb{N}\) be the number of sequential black-box uses in one trial, and let \(L \in \mathbb{N}\) be the number of independent trials. The total query complexity is therefore \(N=mL\).

We examine two distinct scenarios within this framework:
\begin{enumerate}
    \item Parallel black-box access (Fig.~\ref{fig:combined}\ref{fig:combined:b}): 
    Each trial applies the black box once to the input state, followed by a measurement. 
    This corresponds to \(m=1\), and hence \(N=L\). We obtain standard-testing and tolerant testing upper bounds using Algorithms~\ref{alg:identity_testing_m=1} and \ref{alg:identity_testing_2} respectively.

    \item Sequential Black-Box access (Fig.~\ref{fig:combined}\ref{fig:combined:c}): Each trial prepares $\rho$ and applies the black box $m$ times sequentially, yielding the state \[\mathcal{N}_\lambda^{\circ m}(\rho)=\underbrace{\mathcal N_\lambda\circ\mathcal N_\lambda\circ\cdots\circ
\mathcal N_\lambda\circ\mathcal N_\lambda}_{m\ \mathrm{times}}(\rho),\]
and then
measures using the binary POVM $\{\rho,I-\rho\}$.
Repeating this procedure for \(L\) independent trials gives \(N=mL\) total queries. We obtain a standard-testing upper bound in this access using  Algorithm~\ref{alg:identity_testing_m>=2}. The case $m = 1$ recovers the parallel black-box access model described earlier.
\end{enumerate}
Our Algorithm~\ref{alg:identity_testing_m=1} has been used before for unitary channel certification~\cite{fawzi2023quantum, jeon2025query}. 
Our Algorithm~\ref{alg:identity_testing_m>=2} reduces to it for $m = 1$.

\emph{Results}:
Our results are formalized in three theorems and summarized in Table~\ref{table 1}.
The standard-testing bounds in Theorem~\ref{TH : Standard tester, optimality} follow as a special case of the tolerant-testing bounds in Theorem~\ref{TH : Tol tester, optimality} by setting \(\varepsilon_2=\varepsilon\) and \(\varepsilon_1=0\), although the algorithms achieving the corresponding upper bounds differ. On the other hand, the sequential black-box access standard-testing upper bound in Theorem~\ref{TH:m-standard-tradeoff} matches the parallel standard-testing upper bound in Theorem~\ref{TH : Standard tester, optimality}. Thus, allowing \(m\) sequential applications of the black-box does not change the query complexity of the problem. 

\begin{theorem}[Standard Tester, Optimality]\label{TH : Standard tester, optimality} 
For the problem in \emph{\ref{prob:P2.1}}, there exists an ancilla-free, non-adaptive
algorithm using individual two-outcome measurements (Fig.~\ref{fig:combined}\ref{fig:combined:b}) that achieves an upper bound
of $O(1/\varepsilon)$ queries. A matching lower bound of $\Omega(1/\varepsilon)$, up to constant factors, holds even in the ancilla-assisted, adaptive model with arbitrary individual measurements (Fig.~\ref{fig:combined}\ref{fig:combined:a}). 
\end{theorem}

\begin{proof} 
    The required upper bound is established by Algorithm~\ref{alg:identity_testing_m=1} in ancilla-free, non-adaptive access as shown in Fig.~\ref{fig:combined}\ref{fig:combined:b}. Under $H_0$, the tester never outputs $X_k=1$, so the type-I error is zero; that is, $\Pr(X_k=1\mid H_0)=0, ~\forall k\in[L]$. A type-II error occurs only when \(X_1=\cdots=X_L=0\) under hypothesis $H_1$. Bounding the type-II error yields \(\prod_{k=1}^{L}\Pr(X_k=0\mid H_1)\le \delta\) the desired query complexity $O(1/\varepsilon)$. Detailed complexity analysis is provided in Appendix~\ref{app:b1}. 
    
    The corresponding lower bound is tight and is proved in the full-access model described in Fig.~\ref{fig:combined}\ref{fig:combined:a}. Its complexity analysis follows from the tolerant-case lower bound in Appendix~\ref{app:LowerBound}, specialized to \(\varepsilon_1=0\) and \(\varepsilon_2=\varepsilon\); see Th.~\ref{Th:LB_Std}.
\end{proof}

\begin{algorithm}[t]
    \caption{Standard tester in $d_\diamond$}
    \label{alg:identity_testing_m=1}
    \begin{algorithmic}
    \Require $\varepsilon \in (0,1)$, $\delta \ll 1$
    \State Initialize $N \gets \mathcal{O}(\log{(1/\delta)}/\varepsilon)$
    \For{$k \in [N]$}
        \State $\rho \gets |0\rangle\!\langle 0|$ with $\ket{0} \in \mathbb{C}^d$; apply $\mathcal{N}_{\lambda}$
        \State Measure $\mathcal{N}_{\lambda}(|0\rangle\!\langle 0|)$ in $\mathbb{M}=$$\{|0\rangle\!\langle 0|, \mathbb{I} - |0\rangle\!\langle 0|\}$.
        \State Observe $X_k \sim \text{Bern}(1 - \langle 0 | \mathcal{N}_{\lambda}(|0\rangle\!\langle 0|) |0 \rangle )$.
    \EndFor
    \If{$\exists k : X_k = 1$}
        \State \Return $H_1 : d_{\diamond}(\mathcal{N}_{\lambda}, \operatorname{id}) \geq \varepsilon$ 
    \Else
        \State \Return $H_0 : \mathcal{N}_{\lambda} = \operatorname{id}$ 
    \EndIf
    \end{algorithmic}
\end{algorithm}
\begin{theorem}[Sequential standard tester, Upper Bound]
\label{TH:m-standard-tradeoff}
There exists an ancilla-free, non-adaptive tester solving Problem~\emph{\ref{prob:P2.1}} with $m$ sequential black-box applications and $L$ input states (Fig.~\ref{fig:combined}\ref{fig:combined:c}) satisfying:
\begin{align}
L = \Omega\left(\ln(1/\delta)\right)
\quad\text{and}\quad
mL > (2/\varepsilon)\ln(1/\delta).
\end{align}

\end{theorem}

\begin{proof}
The required upper bound is established by Algorithm~\ref{alg:identity_testing_m>=2} in ancilla-free, non-adaptive access as shown in Fig.~\ref{fig:combined}\ref{fig:combined:c}. Under $H_0$, the tester never outputs $X_k=1$, so the type-I error is zero. Under the hypothesis \( H_1 \), we first establish that achieving a target error \( \delta \in (0, 1/2) \) for finite \( d \) necessitates \( L = \Omega(\log(1/\delta)) \). Having established the minimal scaling of \(L\), we next determine the scaling of \(m\) required to keep the type-II error at most \(\delta\). As detailed in Appendix~\ref{app:b2}, any such choice of \((m,L)\), with \(L=\Omega(\log(1/\delta))\), must satisfy:
\begin{align}
N = mL > \frac{2}{\varepsilon}\,\ln\frac{1}{\delta}.
\label{eq:scaling_ml}
\end{align}
This establishes the query complexity as stated in the Th.~\ref{TH:m-standard-tradeoff}.
\end{proof}
To achieve the optimal query complexity of standard testing, there must be a tradeoff between the number of black boxes, $m$, in a sequence and the number of parallel state/measurement pairs, $L$, i.e., the number of trials. Any choice of $(m,L)$  for instance $m=\mathcal{O}(1/\varepsilon)$ and $L=\Omega(\ln(1/\delta))$, or $m=O(1)$ and $L=\mathcal{O}((1/\varepsilon)\ln(1/\delta))$, or any intermediate allocation  achieves the same overall scaling as described in Eqn.~\eqref{eq:scaling_ml}.

\begin{algorithm}[t]
    \caption{Standard tester with sequential black-box access in $d_\diamond$}
    \label{alg:identity_testing_m>=2}
    \begin{algorithmic}
    \Require $\varepsilon \in (0,1)$, $\delta \ll 1$
    \State Initialize $L \gets \log(1/\delta)$, $m \gets \mathcal{O}(1/\varepsilon)$
    \For{$k \in [L]$}
        \State $\rho \gets |0\rangle\!\langle 0|$ with $\ket{0} \in \mathbb{C}^d$; apply $\mathcal{N}_{\lambda}^{\circ m}$
        \State Measure $\mathcal{N}_{\lambda}^{\circ m}(|0\rangle\!\langle 0|)$ in $\mathbb{M}=$$\{|0\rangle\!\langle 0|, \mathbb{I} - |0\rangle\!\langle 0|\}$.
        \State Observe $X_k \sim \text{Bern}(1 - \langle 0 | \mathcal{N}_{\lambda}^{\circ m}(|0\rangle\!\langle 0|) |0 \rangle )$.
    \EndFor
    \If{$\exists k : X_k = 1$}
        \State \Return $H_1 : d_{\diamond}(\mathcal{N}_{\lambda}, \operatorname{id}) \geq \varepsilon$ 
    \Else
        \State \Return $H_0 : \mathcal{N}_{\lambda} = \operatorname{id}$ 
    \EndIf
    \end{algorithmic}
\end{algorithm}

\begin{theorem}[Tolerant Tester, Optimality]\label{TH : Tol tester, optimality}
For the tolerant testing problem in \emph{\ref{prob:P2.2}}, we give an ancilla-free, non-adaptive
algorithm using individual measurements that achieves an upper bound of
$O\bigl(\varepsilon_2/(\varepsilon_2 - \varepsilon_1)^2\bigr)$ queries (Fig.~\ref{fig:combined}\ref{fig:combined:b}). A matching lower
bound of $\Omega \bigl(\varepsilon_2/(\varepsilon_2 - \varepsilon_1)^2\bigr)$, up to
constant factors, holds even in the fully general ancilla-assisted, adaptive model with
arbitrary individual measurements (Fig.~\ref{fig:combined}\ref{fig:combined:a}).
\end{theorem}

\begin{proof} 
  The required upper bound is established by Algorithm~\ref{alg:identity_testing_2} in an ancilla-free, non-adaptive access as shown in Fig.~\ref{fig:combined}\ref{fig:combined:b}. For $t \in [N]$, each use of the channel after measurement yields a Bernoulli outcome, $X_t$, with $\Pr(X_t=1)=\lambda(1-1/d)=p$. After $N$ repetitions, the tester calculates the empirical sum $S=\sum_{t=1}^N X_t$ and return $H_1$ whenever the empirical mean satisfies $S/N\ge\tau$.

Since the binomial upper and lower tail probabilities are monotonic in $p$ (see Lemma~\ref{lem:binomial_monot_lower},~\ref{lem:binomial_monot_upper}), the least-favorable (worst-case) choices of $p$ i.e., those maximizing the respective error probabilities, occur at $p=\varepsilon_1/2c$ for $H_0$ and at $p=\varepsilon_2/2c$ for $H_1$, where $c=(1+1/d)$. Applying the Chernoff--Hoeffding bounds from Lemma~\ref{Chernoff-Hoeffding-Bound} to the binomial tails yields
$\Pr^{\mathrm{err}}_{H_0}\le \exp\big(-N D_{KL}(\tau\|\varepsilon_1/2c)\big)$ and
$\Pr^{\mathrm{err}}_{H_1}\le \exp\big(-N D_{KL}(\tau\|\varepsilon_2/2c)\big)$.
For a threshold $\tau$ to achieve the same target error $\delta$ under both hypotheses, we set the type-I and type-II errors as
$\Pr^{\mathrm{err}}_{H_0}\le e^{-N D_{KL}(\tau\|\varepsilon_1/2c)}\le\delta$ and
$\Pr^{\mathrm{err}}_{H_1}\le e^{-N D_{KL}(\tau\|\varepsilon_2/2c)}\le\delta$ respectively. It suffices to choose $\tau$ that satisfies
$D_{KL}(\tau\|\varepsilon_1/2c)=D_{KL}(\tau\|\varepsilon_2/2c)$. This yields the closed form:
\begin{align}
\tau=
\frac{\ln[(2c-\varepsilon_1)/(2c-\varepsilon_2)]}
     {\ln[\varepsilon_2(2c-\varepsilon_1)/(\varepsilon_1(2c-\varepsilon_2))]}
\in(\varepsilon_1/2c,\varepsilon_2/2c).
\end{align}

Further analysis, as detailed in Appendix~\ref{sec:tol_test_UB}, yields the query complexity 
\begin{align}
N = O\left(\frac{\varepsilon_2\,\ln(1/\delta)}{(\varepsilon_2-\varepsilon_1)^2}\right).
\label{Eqn:UB_N_tol}
\end{align}

The lower bound for Problem~\ref{prob:P2.2} is established in Appendix~\ref{app:LowerBound}, specifically in Theorem~\ref{Th:LB_Tol}, within the full-access model as depicted in Figure~\ref{fig:combined}\ref{fig:combined:a}. We first realize the worst-case channels corresponding to the extremal depolarizing parameter values. From Lemma~\ref{lem:diadistBetDepolChannels}, we have \(\lambda_1 = \varepsilon_1/2v\) and
\(\lambda_2 = \varepsilon_2/2v\), for the two respective hypotheses, where $v=(1-1/d^2)$. The Kullback–Leibler (KL) divergence, $D_{KL}$, between the two joint distributions corresponding to $\lambda_1$ and $\lambda_2$, in the adaptive, ancilla-assisted setting with arbitrary per-trial POVM is lower bounded by a function of $\delta$ using a data-processing inequality.
Explicitly evaluating an upper bound on $D_{KL}$ for the two distributions yields the corresponding lower bound on the query complexity: $N =\Omega({\varepsilon_2\ln(1/\delta)/(\varepsilon_2-\varepsilon_1)^2})$.

This matches, up to constants, the query complexity upper bound achieved by \Cref{alg:identity_testing_2} in Eq.~(\ref{Eqn:UB_N_tol}), establishing optimality.

\end{proof}

\begin{algorithm}[t]
    \caption{Tolerant tester in $d_\diamond$}
    \label{alg:identity_testing_2}
    \begin{algorithmic}[1]
    \Require $0 <\varepsilon_1 <\varepsilon_2 < 1$, $\delta \ll1$ 
    \State Initialize $N \gets 
    \Theta(\varepsilon_2\ln(1/\delta)/(\varepsilon_2 - \varepsilon_1)^2), S \gets 0$
    \State Initialize $\tau \gets
\frac{\ln[(2c-\varepsilon_1)/(2c-\varepsilon_2)]}
     {\ln[\varepsilon_2(2c-\varepsilon_1)/(\varepsilon_1(2c-\varepsilon_2))]}$, $c \gets (1+1/d)$
    \For{$t \in [N]$}
        \State $\rho \gets |0\rangle\!\langle 0|$ with $\ket{0} \in \mathbb{C}^d$; apply $\mathcal{N}_{\lambda}$
        \State Measure $\mathcal{N}_{\lambda}(|0\rangle\!\langle 0|)$ in $\mathbb{M}=$ $\{|0\rangle\!\langle 0|, \mathbb{I} - |0\rangle\!\langle 0|\}$.
        \State Observe $X_t \sim \text{Bern}(1 - \langle 0 | \mathcal{N}_{\lambda}(|0\rangle\!\langle 0|) |0 \rangle )$.
        \State $S = S+X_t$
    \EndFor
    \If{$S/N \geq \tau$} 
        \State \Return $H_1 : d_{\diamond}(\mathcal{N}_{\lambda}, \operatorname{id}) \geq \varepsilon_2$
    \Else
        \State \Return $H_0 : d_{\diamond}(\mathcal{N}_{\lambda}, \operatorname{id}) \leq \varepsilon_1$
    \EndIf
\end{algorithmic}
\end{algorithm} 

\emph{Conclusion and Discussion}:
We provide fast algorithms for standard and tolerant testing of the depolarizing noise assumption on the implementation of unitary channels.
They are in the non-adaptive, ancilla-free setting with two-outcome individual measurements. 
We show that, allowing sequential black box applications does not improve the query complexity, and establish that to achieve the optimal query complexity of standard testing, a tradeoff must exist between the number of black-boxes, $m$, in a sequence and the number of parallel state/measurements pairs, $L$, used.

We also show that access to additional resources like, an ancilla, adaptive inputs and multi-outcome individual measurements do not help by establishing a lower bound in the \emph{full-access} model matching the ancilla-free, non-adaptive upper bounds up to constants. This is in contrast to testing unitary channels with coherent noise, which scales exponentially even with arbitrarily large ancilla, entanglement, adaptive inputs and measurements~\cite{jeon2025query}.
This highlights an exponential gap between testing for the `easiest' and `hardest' error channels on the path to fault-tolerant quantum computers. Knowing what lies inbetween would be illuminating.

Also important to know would be the necessity of access to $U^\dagger$ for achieving optimal query complexity. 
An alternative reduction from Problem~\ref{prob:pre-reduction} to Problem~\ref{prob:testing-depolarizing} follows the access model of Ref.~\cite{jeon2025query}. Therein, access is available to the noisy implemented channel $\mathcal{N}_{U_*}=\mathcal N_U\circ\mathcal N_\lambda$ as a black-box as well as $\mathcal N_{U_*^{\dagger}}=\mathcal N_{U^\dagger}\circ\mathcal N_\lambda$, while measurements are performed using the fixed POVM, $\mathbb M=\{M_x\}_x$.
Our results apply directly to this alternate model.
Problem~\ref{prob:pre-reduction} now reduces to distinguishing $\mathcal N_{\gamma} = \mathcal N_{\lambda}^{\circ2}$ --- a depolarizing channel with black-box parameter $\gamma = 1 - (1 - \lambda^2) \in [0,1]$)--- from the identity channel. 
In this situation the guarantees on $d_{\diamond}(\mathcal{N}_{\lambda}, \operatorname{id})$ testing carry over to test $d_{\diamond}(N_{\gamma}, \operatorname{id})$, using the following inequalities:
$2d_{\diamond}(\mathcal{N_{\lambda},\operatorname{id}}) \ge d_{\diamond}(N_{\gamma}, \operatorname{id}) \ge d_{\diamond}(\mathcal{N_{\lambda}},\operatorname{id})$ (refer Lemma~\ref{lem:diamonddistbounds}) and \( \varepsilon_1 < \varepsilon_2/2 \). 

For learning a unitary channel in diamond distance, the lower bound of $\Omega(d^2/\varepsilon)$ persists even when the algorithm is granted access to $U^\dagger$ and controlled-$U$~\cite{haah2023query}.
So access to $U^\dagger$ offers no advantage.
For testing unitary channels, our complexities as well as those of $\Theta(d/\varepsilon^2)$ in the incoherent setting \cite{fawzi2023quantum} and $\Theta(\sqrt{d}/\varepsilon)$ in the coherent setting \cite{ jeon2025query} are established
under access to both $U$ and $U^\dagger.$
Knowing the complexities under $U$-only access has important practical consequences.

\begin{acknowledgments}
We thank Matthias Caro, Sitan Chen, and Omar Fawzi for fruitful discussions, and Frederik vom Ende for pointing to the SDP formulations of the diamond distance in Ref.~\cite{khatri2020principles}. This work was supported, in part, by the 
EPSRC New Horizons grant EP/X018180/1, and the
Hub for Quantum Computing via Integrated and Interconnected Implementations (QCI3) (EP/Z53318X/1).
\end{acknowledgments}

\bibliographystyle{apsrev4-2}
\bibliography{references}

@article{surawy2022projected,
  doi = {10.22331/q-2022-10-20-844},
  url = {https://doi.org/10.22331/q-2022-10-20-844},
  title = {Projected {L}east-{S}quares {Q}uantum {P}rocess {T}omography},
  author = {Surawy-Stepney, Trystan and Kahn, Jonas and Kueng, Richard and Guta, Madalin},
  journal = {{Quantum}},
  issn = {2521-327X},
  publisher = {{Verein zur F{\"{o}}rderung des Open Access Publizierens in den Quantenwissenschaften}},
  volume = {6},
  pages = {844},
  month = oct,
  year = {2022}
}

@misc{khatri2020principles,
  title={Principles of Quantum Communication Theory: A Modern Approach}, 
      author={Sumeet Khatri and Mark M. Wilde},
      year={2024},
      eprint={2011.04672},
      archivePrefix={arXiv},
      primaryClass={quant-ph},
      url={https://arxiv.org/abs/2011.04672}, 
}

@inproceedings{haah2023query,
  title={Query-optimal estimation of unitary channels in diamond distance},
  author={Haah, Jeongwan and Kothari, Robin and O’Donnell, Ryan and Tang, Ewin},
  booktitle={2023 IEEE 64th Annual Symposium on Foundations of Computer Science (FOCS)},
  pages={363--390},
  year={2023},
  organization={IEEE},
  url={https://doi.org/10.1109/FOCS57990.2023.00028}
}

@misc{rosenthal2024quantum,
  title={Quantum Channel Testing in Average-Case Distance}, 
      author={Gregory Rosenthal and Hugo Aaronson and Sathyawageeswar Subramanian and Animesh Datta and Tom Gur},
      year={2024},
      eprint={2409.12566},
      archivePrefix={arXiv},
      primaryClass={quant-ph},
      url={https://arxiv.org/abs/2409.12566}, 
}

@inproceedings{fawzi2023quantum,
  title = 	 {Quantum Channel Certification with Incoherent Measurements},
  author =       {Fawzi, Omar and Flammarion, Nicolas and Garivier, Aur{\'e}lien and Oufkir, Aadil},
  booktitle = 	 {Proceedings of Thirty Sixth Conference on Learning Theory},
  pages = 	 {1822--1884},
  year = 	 {2023},
  editor = 	 {Neu, Gergely and Rosasco, Lorenzo},
  volume = 	 {195},
  series = 	 {Proceedings of Machine Learning Research},
  month = 	 {12--15 Jul},
  publisher =    {PMLR},
  url = 	 {https://proceedings.mlr.press/v195/fawzi23a.html}
}

@article{sasaki2002optimal,
  title={Optimal parameter estimation of a depolarizing channel},
  author={Sasaki, Masahide and Ban, Masashi and Barnett, Stephen M},
  journal={Physical Review A},
  volume={66},
  number={2},
  pages={022308},
  year={2002},
  publisher={APS},
  url={https://doi.org/10.1103/PhysRevA.66.022308}
}

@article{ji2008parameter,
  title={Parameter estimation of quantum channels},
  author={Ji, Zhengfeng and Wang, Guoming and Duan, Runyao and Feng, Yuan and Ying, Mingsheng},
  journal={IEEE Transactions on Information Theory},
  volume={54},
  number={11},
  pages={5172--5185},
  year={2008},
  publisher={IEEE},
  url={https://doi.org/10.1109/TIT.2008.929940}
}

@article{FawziOufkirFranca2025,
  author  = {Fawzi, Omar and Oufkir, Aadil and Stilck França, Daniel},
  journal = {IEEE Transactions on Information Theory},
  title   = {Lower Bounds on Learning Pauli Channels With Individual Measurements},
  year    = {2025},
  volume  = {71},
  number  = {4},
  pages   = {2642-2661},
  doi = {10.1109/TIT.2025.3527902}
}

@misc{wu2020itstats,
  author       = {Yihong Wu},
  title        = {Lecture Notes on Information-Theoretic Methods for High-Dimensional Statistics},
  year         = {2020},
  url = {http://www.stat.yale.edu/~yw562/teaching/it-stats.pdf}
}

@article{chen2024tight,
  title={Tight bounds on Pauli channel learning without entanglement},
  author={Chen, Senrui and Oh, Changhun and Zhou, Sisi and Huang, Hsin-Yuan and Jiang, Liang},
  journal={Physical Review Letters},
  volume={132},
  number={18},
  pages={180805},
  year={2024},
  publisher={APS},
  url={https://doi.org/10.1103/PhysRevLett.132.180805}
}

@article{chen2025efficient,
  title={Efficient Pauli channel estimation with logarithmic quantum memory},
  author={Chen, Sitan and Gong, Weiyuan},
  journal={PRX Quantum},
  volume={6},
  number={2},
  pages={020323},
  year={2025},
  publisher={APS},
  url={https://doi.org/10.1103/PRXQuantum.6.020323}
}

@misc{gerbessiotis2025survey,
  title={A survey of Chernoff and Hoeffding bounds}, 
      author={Alexandros V. Gerbessiotis},
      year={2025},
      eprint={2506.15612},
      archivePrefix={arXiv},
      primaryClass={cs.DM},
      url={https://arxiv.org/abs/2506.15612}, 
}

@INPROCEEDINGS{oufkir2023sample,
  author={Oufkir, Aadil},
  booktitle={2023 IEEE International Symposium on Information Theory (ISIT)}, 
  title={Sample-Optimal Quantum Process Tomography with non-adaptive Incoherent Measurements}, 
  year={2023},
  volume={},
  number={},
  pages={1919-1924},
  doi={10.1109/ISIT54713.2023.10206538}}

@misc{mele2025optimal,
  title={Optimal learning of quantum channels in diamond distance}, 
      author={Antonio Anna Mele and Lennart Bittel},
      year={2026},
      eprint={2512.10214},
      archivePrefix={arXiv},
      primaryClass={quant-ph},
      url={https://arxiv.org/abs/2512.10214}, 
}

@article{sblg-fbq4,
  title = {Fault-Tolerant Stabilizer Measurements in Surface Codes with Three-Qubit Gates},
  author = {Old, Josias and Tasler, Stephan and Hartmann, Michael J. and M\"uller, Markus},
  journal = {Phys. Rev. Lett.},
  volume = {135},
  issue = {24},
  pages = {240601},
  numpages = {8},
  year = {2025},
  month = {Dec},
  publisher = {American Physical Society},
  doi = {10.1103/sblg-fbq4},
  url = {https://link.aps.org/doi/10.1103/sblg-fbq4}
}

@article{PRXQuantum.5.030352,
  title = {Improving Threshold for Fault-Tolerant Color-Code Quantum Computing by Flagged Weight Optimization},
  author = {Takada, Yugo and Fujii, Keisuke},
  journal = {PRX Quantum},
  volume = {5},
  issue = {3},
  pages = {030352},
  numpages = {16},
  year = {2024},
  month = {Sep},
  publisher = {American Physical Society},
  doi = {10.1103/PRXQuantum.5.030352},
  url = {https://link.aps.org/doi/10.1103/PRXQuantum.5.030352}
}

@article{jeon2025query,
  title   = {On the query complexity of unitary channel certification},
  author  = {Jeon, Sangwoo and Oh, Changhun},
  journal = {npj Quantum Information},
  volume  = {12},
  number  = {1},
  pages   = {2},
  year    = {2025},
  doi     = {10.1038/s41534-025-01135-5},
  url     = {https://doi.org/10.1038/s41534-025-01135-5}
}

@misc{scholz2008confidence,
  title={Confidence bounds and intervals for parameters relating to the binomial, negative binomial, Poisson and hypergeometric distributions with applications to rare events},
  author={Scholz, Fritz},
  year={2008},
  institution={University of Washington},
  url={https://faculty.washington.edu/fscholz/DATAFILES/ConfidenceBounds.pdf}
}

@article{Proctor2025Benchmarking,
  title        = {Benchmarking quantum computers},
  author       = {Proctor, Timothy and Young, Kevin and Baczewski, Andrew D. and Blume-Kohout, Robin},
  journal      = {Nature Reviews Physics},
  year         = {2025},
month = {Feb},
  url          = {https://www.nature.com/articles/s42254-024-00796-z},
volume = {7},
number = {2},
pages={105-118},
issn={2522-5820}
}

@Article{Malhotra2024,
author={Malhotra, Pranit
and Kumar, Ajay
and Garhwal, Sunita},
title={A Systematic Review of Quantum BenchMarking},
journal={International Journal of Theoretical Physics},
year={2024},
month={Oct},
day={30},
volume={63},
number={11},
pages={278},
issn={1572-9575},
url={https://doi.org/10.1007/s10773-024-05811-8}
}

@misc{rohe2025quantumcomputerbenchmarkingexplorative,
      title={Quantum Computer Benchmarking: An Explorative Systematic Literature Review}, 
      author={Tobias Rohe and Federico Harjes Ruiloba and Sabrina Egger and Sebastian von Beck and Jonas Stein and Claudia Linnhoff-Popien},
      year={2025},
      eprint={2509.03078},
      archivePrefix={arXiv},
      primaryClass={quant-ph},
      url={https://arxiv.org/abs/2509.03078} 
}

@misc{chen2026strict,
  title={Strict Hierarchy for Quantum Channel Certification to Unitary}, 
      author={Kean Chen and Qisheng Wang and Zhicheng Zhang},
      year={2026},
      eprint={2604.26900},
      archivePrefix={arXiv},
      primaryClass={quant-ph},
      url={https://arxiv.org/abs/2604.26900}, 
}

@article{FlammiaWallman2020,
  author  = {Flammia, Steven T. and Wallman, Joel J.},
  title   = {Efficient Estimation of {P}auli Channels},
  journal = {ACM Transactions on Quantum Computing},
  volume  = {1},
  number  = {1},
  pages   = {3:1--3:32},
  year    = {2020},
  doi     = {10.1145/3408039}
}

@article{Erhard2019CycleBenchmarking,
  author  = {Erhard, Alexander and Wallman, Joel J. and Postler, Lukas and Meth, Michael and Stricker, Roman and Martinez, Esteban A. and Schindler, Philipp and Monz, Thomas and Emerson, Joseph and Blatt, Rainer},
  title   = {Characterizing Large-Scale Quantum Computers via Cycle Benchmarking},
  journal = {Nature Communications},
  volume  = {10},
  pages   = {5347},
  year    = {2019},
  doi     = {10.1038/s41467-019-13068-7}
}

@article{HarperFlammiaWallman2020,
  author  = {Harper, Robin and Flammia, Steven T. and Wallman, Joel J.},
  title   = {Efficient Learning of Quantum Noise},
  journal = {Nature Physics},
  volume  = {16},
  pages   = {1184--1188},
  year    = {2020},
  doi     = {10.1038/s41567-020-0992-8}
}

@article{ChenLearnability2023,
  author  = {Chen, Senrui and Liu, Yunchao and Otten, Matthew and Seif, Alireza and Fefferman, Bill and Jiang, Liang},
  title   = {The Learnability of {P}auli Noise},
  journal = {Nature Communications},
  volume  = {14},
  pages   = {52},
  year    = {2023},
  doi     = {10.1038/s41467-022-35759-4}
}

@article{doi:10.1137/S0097539799359385,
author = {Aharonov, Dorit and Ben-Or, Michael},
title = {Fault-Tolerant Quantum Computation with Constant Error Rate},
journal = {SIAM Journal on Computing},
volume = {38},
number = {4},
pages = {1207-1282},
year = {2008},
doi = {10.1137/S0097539799359385},

URL = { 
    
        https://doi.org/10.1137/S0097539799359385
    
    

}}

@article{Sanders_2016,
doi = {10.1088/1367-2630/18/1/012002},
url = {https://doi.org/10.1088/1367-2630/18/1/012002},
year = {2015},
month = {dec},
publisher = {IOP Publishing},
volume = {18},
number = {1},
pages = {012002},
author = {Sanders, Yuval R and Wallman, Joel J and Sanders, Barry C},
title = {Bounding quantum gate error rate based on reported average fidelity},
journal = {New Journal of Physics}
}

@book{gs007,
 author = {Montanaro, Ashley and Wolf, Ronald {de}},
 title = {A Survey of Quantum Property Testing},
 year = {2016},
 pages = {1--81},
 doi = {10.4086/toc.gs.2016.007},
 publisher = {Theory of Computing Library},
 number = {7},
 series = {Graduate Surveys},
 URL = {http://www.theoryofcomputing.org/library.html},
}

@article{Raussendorf_2007,
doi = {10.1088/1367-2630/9/6/199},
url = {https://doi.org/10.1088/1367-2630/9/6/199},
year = {2007},
month = {jun},
publisher = {},
volume = {9},
number = {6},
pages = {199},
author = {Raussendorf, R and Harrington, J and Goyal, K},
title = {Topological fault-tolerance in cluster state quantum computation},
journal = {New Journal of Physics}
}

\clearpage

\appendix

\section{Lemmas}
\label{sec:AppA}

\subsection{Depolarizing Channel}

\begin{applemma}[Unitary Covariance of the Depolarizing Channel] 
Let $\mathcal N_\lambda(\rho)=(1-\lambda)\rho+\lambda\mathbb I/d$ and $\mathcal N_U(\rho)=U\rho U^\dagger$. Then $
\mathcal N_\lambda(U\rho U^\dagger)=U\,\mathcal N_\lambda(\rho)\,U^\dagger,$
and therefore $\mathcal N_U\circ\mathcal N_\lambda=\mathcal N_\lambda\circ\mathcal N_U$.
\label{Dep_channel covariance}
\end{applemma}
\begin{proof}
Using $U\mathbb I U^\dagger=\mathbb I$,
\begin{align}
\mathcal N_\lambda(U\rho U^\dagger)
&=(1-\lambda) U\rho U^\dagger+\lambda\frac{\mathbb I}{d}
\\&=U\left((1-\lambda)\rho+\lambda\frac{\mathbb I}{d}\right)U^\dagger \\&= U \mathcal{N}_{\lambda}(\rho) U^{\dagger}
\end{align}
\end{proof}

\begin{applemma}[Query Complexity of Depolarizing Parameter Estimation]
\label{eq:depol-lemma}
Let $\mathcal{H}\cong\mathbb{C}^d$ and $\rho = \ket{0}\!\bra{0} \in \mathcal{L}(\mathcal H)$ and let $\mathcal{N}_\lambda$ be the depolarizing channel: 
\begin{equation}
    \mathcal{N}_\lambda(\rho) = (1-\lambda)\rho + \lambda \frac{\mathbb{I}}{d}.
\end{equation}
Consider the following estimation procedure repeated $N$ times:
\begin{enumerate}
    \item Prepare the state $\rho$ and apply the channel $\mathcal{N}_\lambda$.
    \item Perform the binary POVM measurement $\{\rho, \mathbb{I} - \rho\}$ to obtain an outcome $X_i \in \{0,1\}$.
\end{enumerate}
Let $\hat{p} = \frac{1}{N} \sum_{i=1}^N X_i$ be the empirical mean and define the estimator:
\begin{equation}
    \hat{\lambda} := \frac{d(1-\hat{p})}{d - 1}.
\end{equation}
Then, for any precision $\varepsilon \in (0,1)$ and failure probability $\delta \in (0,1/2)$, we have:
\begin{align}
    &\Pr
    \left(
    \|\mathcal{N}_{\hat{\lambda}}(\rho)-\mathcal{N}_{\lambda}(\rho)\|_1
    \ge \varepsilon
    \right)
    \\&=
    \Pr
    \left(
    |\hat{\lambda}-\lambda|
    \ge
    \frac{\varepsilon}{2(1-1/d)}
    \right)
    \le \delta.
\end{align}
whenever the number of samples satisfies:
\begin{equation}
    N \ge \frac{2}{\varepsilon^2}\ln\frac{2}{\delta}.
\end{equation}
\end{applemma}

\begin{proof}
The probability of observing outcome $\rho$ (success) in a single trial is given by:
\begin{align}
    p &= \mathrm{Tr}(\rho \mathcal N_\lambda(\rho)) 
    = (1-\lambda)\mathrm{Tr}(\rho^2) + \frac{\lambda}{d}\mathrm{Tr}(\rho) \\
    &= 1-\lambda + \frac{\lambda}{d}
    = 1-\lambda\left(1-\frac{1}{d}\right).
\end{align}
Solving for $\lambda$ yields $\lambda = \frac{d(1-p)}{d-1}$. By linearity, $\mathbb{E}[\hat{p}]=p$, making $\hat\lambda$ an unbiased estimator of $\lambda$.
The estimation error relates to $\hat{p}$ by a constant factor:
\begin{equation}
    \label{eqn:lambdaAndPRelate}
    |\hat\lambda - \lambda| 
    = \left| \frac{d(1-\hat p)}{d-1} - \frac{d(1-p)}{d-1} \right|
    = \frac{d}{d-1}|\hat p - p|.
\end{equation}
We apply Hoeffding's inequality for bounded variables $X_i \in [0,1]$:
\begin{equation}
\label{eqn:hoeffdingsFirstUse}
    \Pr(|\hat p - p| \ge t) \le 2\exp(-2Nt^2).
\end{equation}
Due to Eqn.\eqref{eqn:lambdaAndPRelate}, the condition 
$|\hat\lambda - \lambda| \ge \frac{\varepsilon}{2(1-1/d)}$ is equivalent to 
$|\hat p - p| \ge \frac{\varepsilon}{2}$. Substituting 
$t = \frac{\varepsilon}{2}$ into Eqn.~\eqref{eqn:hoeffdingsFirstUse} and requiring the probability of error be at most $\delta$:
\begin{equation}
    2\exp\left(-2N \left(\frac{\varepsilon}{2}\right)^2\right) \le \delta 
    \iff 
    N \ge \frac{2}{\varepsilon^2}\ln\frac{2}{\delta}.
\end{equation}
\end{proof}

\subsection{Distance Lemmas}
The following distance lemmas will be used in the upper and lower bound calculations in the following sections.
\begin{applemma}[Trace distance between depolarizing channel and the identity channel]
\label{lem:trdistBetDepolChannels}
The distance between $\mathcal{N}_{\lambda}$ and the identity channel, $\operatorname{id}$, in trace distance, is:
\begin{align}
    d_{\operatorname{Tr}}(\mathcal{N}_\lambda, \operatorname{id}) &:= \max_{\rho \in \mathbb{C}^{d \times d}} \left\| (\mathcal{N}_\lambda - \operatorname{id})(\rho) \right\|_1 = 2\lambda(1-\tfrac{1}{d}).
    \label{equality, trace distance}
\end{align}
Further, any pure state maximizes the one norm.
\end{applemma}

\begin{proof}We prove the equality in Eq.~\eqref{equality, trace distance} by matching the lower and upper bounds on the trace distance. From the definition of the trace distance it follows:
    \begin{align}
    d_{\operatorname{Tr}} \left(\mathcal{N}_\lambda, \operatorname{id}\right) &:= \max_{\rho \in \mathbb{C}^{d \times d}} \left\| (1 -\lambda)\rho + \lambda \dfrac{\mathbb{I}}{d} - \rho \right\|_1  
    \\&= \lambda \cdot \max_{\rho \in \mathbb{C}^{d \times d}} \left\|\frac{\mathbb{I}}{d} - \rho \right\|_1.  
    \label{Lemma A.1, 3}
\end{align}
For any pure state $\sigma$, 
\begin{align}
    d_{\operatorname{Tr}} \left(\mathcal{N}_\lambda, \operatorname{id}\right) \geq \lambda \left\|\frac{\mathbb{I}}{d} - \sigma \right\|_1.
    \label{LemmaA.1, 1}  
\end{align}
The Schatten 1-norm of \( \frac{\mathbb{I}}{d} - \sigma \) is given by $\left\|\frac{\mathbb{I}}{d} - \sigma \right\|_1 = \sum_{i=1}^d \Delta_i(\frac{\mathbb{I}}{d} - \sigma) = \sum_{i=1}^d |\zeta_i|$, where \( \Delta_i(\frac{\mathbb{I}}{d} - \sigma) \) and \( \zeta_i \) represent the singular values and eigenvalues of \( \frac{\mathbb{I}}{d} - \sigma \), respectively. Since \( \sigma \) is a pure state, \( \frac{\mathbb{I}}{d} - \sigma\) has eigenvalues: \(\zeta_1 = \frac{1-d}{d} \) and \(\zeta_2, \zeta_3 ... \zeta_d = \tfrac{1}{d} \). Therefore,
\begin{align}
    \bigg \|\frac{\mathbb{I}}{d} - \sigma \bigg \|_1 = \bigg|\frac{1-d}{d}\bigg| + (d-1)\dfrac{1}{d} = 2 \bigg(1-\dfrac{1}{d}\bigg).
    \label{LemmaA.1, 2}
\end{align}
Substituting  Eqn.~\eqref{LemmaA.1, 2} in Eqn.~\eqref{LemmaA.1, 1} we get: 
\begin{align}
    d_{\operatorname{Tr}} \left(\mathcal{N}_\lambda, \operatorname{id}\right) \geq 2 \lambda \bigg(1-\dfrac{1}{d}\bigg).
    \label{lemma A.1, lower}
\end{align}

We begin proving the inequality in the other direction by recalling the semidefinite program (SDP) representation of the trace norm from Ref.~\cite{khatri2020principles}: for any Hermitian operator $H$, its trace norm is given by:
\begin{align}
\label{eq:trace-norm-sdp}
\|H\|_{1}
\,=\,
\inf_{\substack{Y_{1},\,Y_{2} \ge 0}}
\left\{
    \operatorname{Tr}[Y_{1} + Y_{2}]
    :\;
    Y_{1} \ge H,\;
    Y_{2} \ge -H
\right\}.
\end{align}
We have already established that the trace distance between $\mathcal{N}_\lambda$ and $\operatorname{id}$ is $d_{\operatorname{Tr}}\left(\mathcal{N}_\lambda,\operatorname{id}\right) =
\lambda \max_{\rho}
\left\| \frac{\mathbb{I}}{d} - \rho \right\|_1.$ For each fixed state $\rho$, define $H := \frac{\mathbb{I}}{d} - \rho.$ Applying the SDP in Eqn.~\eqref{eq:trace-norm-sdp} yields:
\begin{align}
   &\left\| \frac{\mathbb{I}}{d} - \rho \right\|_1
\nonumber\\&=
\inf_{\substack{Y_{1},\,Y_{2} \ge 0}}
\left\{
    \operatorname{Tr}[Y_{1} + Y_{2}]
    :\;
    Y_{1} \ge \frac{\mathbb{I}}{d} - \rho,\;
    Y_{2} \ge \rho - \frac{\mathbb{I}}{d}
\right\}.
\end{align}

To obtain an upper bound on the trace norm, we choose a convenient feasible pair \[Y_{1} := \frac{1}{d}(\mathbb{I} - \rho),\quad
Y_{2} := \left(1 - \frac{1}{d}\right)\rho.\] Both $Y_{1}$ and $Y_{2}$ are positive semidefinite, and one verifies its feasibility in Eqn.~\eqref{eq:trace-norm-sdp} by checking : $Y_{1} - \left(\frac{\mathbb{I}}{d} - \rho\right)
= \left(1 - \frac{1}{d}\right)\rho \ge 0 \text{  and  }
Y_{2} - \left(\rho - \frac{\mathbb{I}}{d}\right)
= \frac{1}{d}(\mathbb{I} - \rho) \ge 0.$ Since $\operatorname{Tr}(Y_{1} + Y_{2})
= 2\left(1 - \frac{1}{d}\right)$, we obtain an upper bound on the $1-$norm which is independent of $\rho$. Thus, maximization of $\left\| \frac{\mathbb{I}}{d} - \rho \right\|_1$ over all $\rho \in \mathcal{L}(\mathbb{C}^d)$ is upper bounded by $2\left(1 - \frac{1}{d}\right)$, which from the definition of the trace distance in Eqn. \eqref{Lemma A.1, 3} implies:
\begin{align}
d_{\operatorname{Tr}}\left(\mathcal{N}_\lambda,\operatorname{id}\right)
&\le
2\lambda \left(1 - \frac{1}{d}\right).
\label{upper_bound_Dtr}
\end{align}
Therefore, from Eqn.~\eqref{lemma A.1, lower} and Eqn.~\eqref{upper_bound_Dtr}, for any pure state $\rho \in \mathcal{L}(\mathbb{C}^d)$:
\begin{align}
    d_{\operatorname{Tr}} \left(\mathcal{N}_\lambda, \operatorname{id}\right) = 2 \lambda \bigg(1-\frac{1}{d}\bigg)
    \label{lemma A.1, equal}
\end{align}
\end{proof}
\begin{applemma}[Diamond distance between depolarizing channel and the identity channel]
\label{lem:diadistBetDepolChannels}
    The diamond distance between the depolarizing channel, $\mathcal{N}_\lambda$, and the identity channel is:
    \begin{align}
        d_{\diamond}(\mathcal{N}_\lambda, \operatorname{id}) &:= \max_{\rho \in \mathbb{C}^{d \times d} \otimes \mathbb{C}^{d \times d}} \left\| \operatorname{id} \otimes (\mathcal{N}_\lambda - \operatorname{id})(\rho) \right\|_1 \\
        &= 2\lambda(1-\tfrac{1}{d^2}).
    \end{align} 
    Further the maximization of the $1-$norm is achieved by the maximally entangled state 
$\ket{\Phi} = \frac{1}{\sqrt{d}}\sum_{i=1}^d \ket{i}_R \ket{i}_A$.
\end{applemma}

 \begin{proof}
We prove matching lower and upper bounds.

From the definition of the diamond distance, we have: 
\begin{align}
d_{\diamond}(\mathcal{N}_\lambda, \operatorname{id})
= \max_{\rho_{RA}}
\bigl\|(\mathcal{N}_\lambda \otimes \mathrm{id})(\rho_{RA})
      - \rho_{RA}\bigr\|_1,
\end{align}
where the maximization of $\rho_{RA}$ is over all density operators on $\mathbb{C}^d \otimes \mathbb{C}^d$.
Thus, \emph{any} choice of $\rho_{RA}$ substituted into $\bigl\|(\mathcal{N}_\lambda \otimes \mathrm{id})(\rho_{RA}) - \rho_{RA}\bigr\|_1$ yields a lower bound on the diamond norm. The choice we use herein is $\rho = \ket{\Phi}\!\bra{\Phi},$
where $\ket{\Phi}
    =
    \frac{1}{\sqrt{d}}
    \sum_{i=1}^{d}
    \ket{i}_{R}\ket{i}_{A}$
is the maximally entangled state on $RA$.
Then,
\begin{align}
(\mathcal{N}_\lambda \otimes \mathrm{id})(\rho) - \rho
= \lambda\Bigl(\frac{\mathbb{I} \otimes \mathbb{I}}{d^2} - \rho\Bigr).
\end{align}
Define
\begin{align}
\mathcal{X} := \frac{\mathbb{I}\otimes\mathbb{I}}{d^2} - \ket{\Phi}\!\bra{\Phi}.
\end{align}
Then, as \(\mathcal{X}\ket{\Phi} = -\bigl(1-\tfrac{1}{d^2}\bigr)\ket{\Phi}\) and
\(\mathcal{X}\ket{\psi} = \tfrac{1}{d^2}\ket{\psi}\) for all \(\ket{\psi}\perp\ket{\Phi}\), the eigenvalues of \(\mathcal{X}\) are:
\begin{align}
-\Bigl(1-\frac{1}{d^2}\Bigr) \text{ and }
\frac{1}{d^2},
\end{align}
with multiplicity $1$ and $d^2-1$, respectively.
Hence, for \(\lambda \ge 0\):
\begin{align}
\label{eqn:XOneNorm}
\|\lambda \mathcal{X}\|_1
= \lambda\Bigl(1-\frac{1}{d^2}\Bigr)
  + (d^2-1)\frac{\lambda}{d^2}
= 2\lambda\Bigl(1-\frac{1}{d^2}\Bigr).
\end{align}

Therefore, we have derived the following lower bound on $d_{\diamond}(\mathcal{N}_\lambda, \operatorname{id})$:
\begin{align}
d_{\diamond}(\mathcal{N}_\lambda, \operatorname{id})
\ge 2\lambda\Bigl(1-\frac{1}{d^2}\Bigr),
\end{align}
and, as shown in Eqn.~\eqref{eqn:XOneNorm}, this one norm is achieved by  $\ket{\Phi} = \frac{1}{\sqrt{d}}\sum_{i=1}^d \ket{i}_R \ket{i}_A$.

Let $\Gamma^{\mathcal{N}_\lambda}:=(\mathrm{id}_R\otimes \mathcal{N}_\lambda)(|\Phi\rangle\!\langle\Phi|_{RA})$ and $\Gamma^{\mathrm{id}}:=(\mathrm{id}_R\otimes \mathrm{id}_A)(|\Phi\rangle\!\langle\Phi|_{RA})=|\Phi\rangle\!\langle\Phi|_{RA},$ be the Choi operators of $\mathcal{N}_{\lambda}$ and $\operatorname{id}$ where $|\Phi\rangle_{RA}:=\sum_i |i\rangle_R\otimes |i\rangle_A.$ is the unnormalized maximally entangled state.
Define the Choi difference, $\Delta$, by:
\begin{align}
\Delta
:= \Gamma^{\mathcal{N}_\lambda} - \Gamma^{\mathrm{id}}.
\end{align}

From the SDP characterization of the diamond norm in \cite{khatri2020principles}, we have
\begin{equation}
    d_{\diamond}(\mathcal{N}_\lambda, \operatorname{id})
= \inf_{\substack{Z \ge 0 \\ Z \ge \Delta}}2
\bigl\|\operatorname{Tr_2}[Z]\bigr\|_\infty.
\label{SDP_diamond}
\end{equation}

We know, \begin{align}\Gamma^{\mathcal{N}_\lambda}
= (1-\lambda)\,\Gamma^{\mathrm{id}}
  + \lambda\,\frac{1}{d}\,\mathbb I \otimes \mathbb I.\end{align} 
  Therefore,
  \begin{align}\Delta
= \lambda\Bigl(\frac{1}{d} \mathbb I \otimes \mathbb I - \Gamma^{\mathrm{id}}\Bigr).
\label{eqn:deltaDef}
\end{align}

Let $P_\parallel := \frac{1}{d}\ket{\Omega}\bra{\Omega},$ and $
P_\perp := \mathbb I \otimes \mathbb I - P_\parallel.$ Then $P_\parallel$ is the rank-one projector onto the (normalized)
maximally entangled state, and $P_\parallel + P_\perp = \mathbb I \otimes \mathbb I$. \newline We have $\Gamma^{\mathrm{id}} = \ket{\Omega}\bra{\Omega}
= d\,P_\parallel$ and $\frac{1}{d} \mathbb I \otimes \mathbb I
= \frac{1}{d}(P_\parallel + P_\perp),$ so, using Eqn.~\ref{eqn:deltaDef}:
\begin{align}
\Delta
= \lambda\Bigl(\frac{1}{d} - d\Bigr) P_\parallel
  + \lambda\frac{1}{d} P_\perp.
\end{align}

From here on, we consider only operators of the form:
\begin{align}
Z = \lambda\bigl(\alpha P_\parallel + \beta P_\perp\bigr),
\qquad \alpha,\beta \ge 0.
\end{align}
For operators of this form, the condition $Z \ge \Delta$ is equivalent to:
\begin{align}
\alpha \;\ge\; \frac{1}{d} - d,
\qquad
\beta \;\ge\; \frac{1}{d}.
\end{align}
Since $\frac{1}{d} - d < 0$, any $\alpha \ge 0$ satisfies the first
inequality, so the only nontrivial constraint is $\beta \ge 1/d$. We now compute the partial trace $\operatorname{Tr}_A[Z]$.
Using $\operatorname{Tr}_A[P_\parallel] = \frac{1}{d} \mathbb I$
and $\operatorname{Tr}_A[\mathbb I \otimes \mathbb I] = d\,\mathbb I$, we obtain:
\begin{align}
\operatorname{Tr}_A[P_\perp]
= \operatorname{Tr}_[\mathbb I \otimes \mathbb I - P_\parallel]
= d\mathbb I - \frac{1}{d} \mathbb I
= \Bigl(d - \frac{1}{d}\Bigr) \mathbb I.
\end{align}
Thus,
\begin{align}
\operatorname{Tr}_A[Z]
= \lambda\Bigl(\alpha\,\frac{1}{d} \mathbb I
               + \beta\Bigl(d - \frac{1}{d}\Bigr)\mathbb I\Bigr)
= \lambda\Bigl(\frac{\alpha}{d}
               + \beta\Bigl(d - \frac{1}{d}\Bigr)\Bigr) \mathbb I.
\end{align}
Since $\operatorname{Tr}_A[Z]$ is a scalar multiple of the identity, its operator norm is just the absolute value of:
\begin{align}
\label{eqn:infNormTraceZ}
\bigl\|\operatorname{Tr}_A[Z]\bigr\|_\infty
= \lambda\Bigl(\frac{\alpha}{d}
               + \beta\Bigl(d - \frac{1}{d}\Bigr)\Bigr).
\end{align}
Setting $\alpha = 0$ and $\beta = 1/d$ in Eqn.~\eqref{eqn:infNormTraceZ} minimizes it subject to the conditions $\alpha \ge 0$ and $\beta \ge 1/d$. Therefore, the minimum possible value of Eqn.~\eqref{eqn:infNormTraceZ} is:
\begin{align}
\label{eqn:infNormTraceZ_2}
\lambda\Bigl(\frac{1}{d}\Bigl(d - \frac{1}{d}\Bigr)\Bigr)
= \lambda\Bigl(1 - \frac{1}{d^2}\Bigr).
\end{align}
Eqn.~\eqref{eqn:infNormTraceZ_2} can then be applied in Eqn.~\eqref{SDP_diamond} to conclude that:
\begin{align}
d_{\diamond}(\mathcal{N}_\lambda, \operatorname{id})
\le 2\lambda\Bigl(1 - \frac{1}{d^2}\Bigr).
\end{align}
Combining this upper bound with the lower bound found previously, we obtain:
\begin{align}
d_{\diamond}(\mathcal{N}_\lambda, \operatorname{id})
= 2\lambda\Bigl(1 - \frac{1}{d^2}\Bigr),
\label{diamond_equal}
\end{align}
as claimed.
\end{proof}
\begin{applemma}[Lower bound on the $d_{\operatorname{Tr}}$] 
\begin{align}
    d_{\operatorname{Tr}}(\mathcal{N}_{\lambda},\operatorname{id}) \geq \dfrac{d_{\diamond}(\mathcal{N}_{\lambda}, \text{id})}{1.5}, \text{ for } d\ge 2.
\end{align}
\end{applemma}
\begin{proof}
Assume there exists a $c \in \mathbb{R}$ such that:
\begin{align}
    \label{eqn:distanceRelationSeeking}
    d_{\operatorname{Tr}}(\mathcal{N}_{\lambda},\operatorname{id}) \geq \dfrac{d_{\diamond}(\mathcal{N}_{\lambda}, \text{id})}{c} \text{ for } d\ge 2.
\end{align}
Then, using Lemmas~\ref{lem:trdistBetDepolChannels}~and~\ref{lem:diadistBetDepolChannels}, this requires:
\begin{align}
\label{eqn:cExistsRequirements}
    1 - \dfrac{1}{d} \geq \dfrac{1}{c} \left ( 1 - \dfrac{1}{d^2} \right ).
\end{align}
Solving Eqn.~\eqref{eqn:cExistsRequirements} for $c$ gives:
\begin{align}
\label{eqn:cExistsRequirementsModified}
    c \geq \dfrac{ 1 - d^{-2}}{1 - d^{-1}}
    =
    \dfrac{d + 1}{d}.
\end{align}
As $(d+1)/d$ is monotonically decreasing, any value of $c$ that satisfies Eqn.~\eqref{eqn:cExistsRequirementsModified} for $d = 2$ satisfies it for any $d \geq 2$. Hence, Eqn.~\eqref{eqn:distanceRelationSeeking} is satisfied -- for all $d \geq 2$ -- by any $c$ such that $c \geq 3/2$. Therefore,
\begin{align}
    d_{\operatorname{Tr}}(\mathcal{N}_{\lambda},\operatorname{id}) \geq \dfrac{d_{\diamond}(\mathcal{N}_{\lambda}, \text{id})}{1.5} \text{ for } d\ge 2.
\end{align}
\end{proof}

\begin{applemma}[Bounds on the $d_{\diamond}$] 
\label{lem:diamonddistbounds}
\begin{align}
   2\cdot d_{\diamond}(\mathcal{N}_{\lambda}, \operatorname{id}) \ge d_{\diamond}(\mathcal{N}^{\circ 2}_{\lambda},\operatorname{id}) \geq d_{\diamond}(\mathcal{N}_{\lambda},\operatorname{id}) \text{ for } d\ge 2
\end{align}
\end{applemma}
\begin{proof}
Assume there exists a $c \in \mathbb{R}$ such that:
\begin{align}
    \label{eqn:diamonddistanceRelationSeeking}
    c~\cdot d_{\diamond}(\mathcal{N}_{\lambda}, \operatorname{id}) \ge d_{\diamond}(\mathcal{N}^{\circ 2}_{\lambda},\operatorname{id}) \text{ for } d\ge 2.
\end{align}
Then, using Lemma~\ref{lem:diadistBetDepolChannels}, this requires:
\begin{align}
\label{eqn:diamondcExistsRequirements}
    c\lambda \ge 1-(1-\lambda)^2.
\end{align}
Solving Eqn.~\eqref{eqn:diamondcExistsRequirements} for $c$ gives:
\begin{align}
\label{eqn:diamondcExistsRequirementsModified}
    c \geq \dfrac{ 1 - (1-\lambda^2)}{\lambda}
\end{align}
The value of $c$ that satisfies Eqn.~\eqref{eqn:diamondcExistsRequirementsModified} for all values of $\lambda \in [0,1]$ is 2. The second inequality follows from the fact that $1-(1-\lambda)^2 \ge \lambda$ for all values of $\lambda \in [0,1]$.
\end{proof}

\subsection{Other Important Lemmas}
These lemmas will be used in the tolerant tester upper bound. Let $n \in \mathbb{N}$ and $k \in \{0, \dots, n\}$. For $p \in [0,1]$, let $S^{(p)} \sim \mathrm{Binomial}(n,p)$. 
\begin{applemma}[Monotonicity of Binomial Lower Tail \cite{scholz2008confidence}]
\label{lem:binomial_monot_lower}
The cumulative distribution function $\operatorname{Pr}(S^{(p)} \le k)$ is nonincreasing in $p$. Specifically, for any $0 \le p_1 \le p_2 \le 1$,
\begin{equation}
    \operatorname{Pr}(S^{(p_2)} \le k) \le \operatorname{Pr}(S^{(p_1)} \le k).
\end{equation}
\end{applemma}

\begin{applemma}[Monotonicity of Binomial Upper Tail \cite{scholz2008confidence}]
\label{lem:binomial_monot_upper}
The probability $\operatorname{Pr}(S^{(p)} \ge k)$ is nondecreasing in $p$. Specifically, for any $0 \le p_1 \le p_2 \le 1$,
\begin{equation}
    \operatorname{Pr}(S^{(p_1)} \ge k) \le \operatorname{Pr}(S^{(p_2)} \ge k).
\end{equation}
\end{applemma}

\begin{applemma}[Chernoff--Hoeffding Bounds 
\label{Chernoff-Hoeffding-Bound}
\cite{gerbessiotis2025survey}]
Let $X_1, \ldots, X_N$ be independent Bernoulli random variables with parameter $p \in (0,1)$. Let $S_N = \sum_{i=1}^N X_i$ denote their sum and $\bar{X} = S_N/N$ denote the empirical mean, such that $\mathbb{E}[\bar{X}] = p$.

For any $r \in (0,1)$, let $D(r \| p)$ denote the Kullback--Leibler (KL) divergence between Bernoulli distributions with parameters $r$ and $p$:
\begin{equation}
    D_{KL}(r \| p) = r \ln \left(\frac{r}{p}\right) + (1-r) \ln \left(\frac{1-r}{1-p}\right).
\end{equation}
Then, the following tail bounds hold:
\begin{align}
    \operatorname{Pr}(\bar{X} \ge r) &\le \exp\big(-N D_{KL}(r \| p)\big) \quad \text{for } p < r < 1, \\
    \operatorname{Pr}(\bar{X} \le r) &\le \exp\big(-N D_{KL}(r \| p)\big) \quad \text{for } 0 < r < p.
\end{align}
\end{applemma}
\section{Upper Bounds}
\subsection{Standard Tester : When sequential black box applications are not allowed} 
\label{app:b1}
\begin{apptheorem}
    For the problem in \emph{\ref{prob:P2.1}}, there exists an ancilla-free, non-adaptive
algorithm using individual two-outcome measurements that achieves an upper bound
of $O(1/\varepsilon)$ queries. 
\end{apptheorem}
\begin{proof}
    From Algorithm~\ref{alg:identity_testing_m=1}, the probability of error under $H_0$ is $0$, since, $\forall k \in [L]$, $\Pr(X_k=1|{H}_0)=0$. Under $H_1$, the error probability is given by (using the independence of each $X_k$):
\begin{align}
    \operatorname{Pr}_{H_1}^{\text{err}} &= \operatorname{Pr}\left( \forall k \in [L] : X_k = 0 \hspace{+1pt}|\hspace{+1pt}{H}_0 \right)
    \\& = {\textstyle \prod_{k=1}^{L} \operatorname{Pr}(X_k = 0 \hspace{+1pt}|\hspace{+1pt}{H}_1)} = \textstyle {\prod_{k=1}^{L} \langle 0 | \mathcal{N}_{\lambda}(|0\rangle\!\langle 0|) | 0\rangle}  \\&= \{1-\lambda(1-\tfrac{1}{d})\}^L \leq \left(1 - \dfrac{\varepsilon}{3}\right)^L \leq {\textstyle \exp \left( -\frac{{\varepsilon} L}{3} \right)} \leq \delta.
    \label{Pr_H1_m=1}
\end{align} 
Where \(\lambda(1-1/d)\ge \varepsilon/3\) for $d\ge2$. Therefore, when sequential black box applications are disallowed, the number of queries sufficient to solve Problem~\ref{prob:P2.1} is $L \geq \frac{3  \log(\frac{1}{\delta})}{\varepsilon}$.

\end{proof}
\subsection{Standard Tester: When sequential black box applications are allowed}
\label{app:b2}

\begin{apptheorem}
    There exists an ancilla-free, non-adaptive tester solving Problem~\emph{\ref{prob:P2.1}} with $m$ sequential black-box applications and $L$ input states (Fig.~\ref{fig:combined}\ref{fig:combined:c}) satisfying:
\begin{align*}
L = \Omega\left(\ln(1/\delta)\right)
\quad\text{and}\quad
mL > \tfrac{2}{\varepsilon}\,\ln\tfrac{1}{\delta}.
\end{align*}
\end{apptheorem}
\begin{proof}
    From Algorithm~\ref{alg:identity_testing_m>=2}, the probability of error under \(H_0\) is \(0\), since:
\begin{align}
\Pr(X_k = 1 \mid H_0) = 0 \quad \forall\,k \in [L].
\end{align}
Under \(H_1\), the error probability is (using independence of the \(X_k\)) given by:
\begin{align}
    \operatorname{Pr}_{H_1}^{\text{err}}
    &= \Pr\left( \forall k \in [L] : X_k = 0 \,\middle|\, H_1 \right) \\
    &= \prod_{k=1}^{L} \Pr(X_k = 0 \mid H_1) 
    = \prod_{k=1}^{L} \bigl\langle 0 \big| \mathcal{N}_{\lambda}^{\circ m}\bigl(\lvert 0\rangle\!\langle 0\rvert\bigr) \big| 0 \bigr\rangle \\
    &= \Bigl\{(1-\lambda)^m\Bigl(1-\tfrac{1}{d}\Bigr)+\tfrac{1}{d}\Bigr\}^L.
\end{align}
Then define the notation:
\begin{align}
A := (1-\lambda)^m\Bigl(1-\frac{1}{d}\Bigr) + \frac{1}{d},
\qquad d \ge 2,
\end{align}
Under $H_1$ we know from Lemma~\ref{lem:diadistBetDepolChannels} that
:
\begin{align}
\lambda \ge \lambda_* := \frac{\varepsilon}{2\left(1-\frac{1}{d^2}\right)},
\qquad \varepsilon \in (0,1).
\end{align}
Since for all real $x$ and finite $m \ge 1$ we have\; \[(1-\lambda)^m < e^{-m\lambda} \le e^{-m\lambda_*}.\]
Where the strict inequality follows from the fact that $\lambda >0$ and $m$ is finite. Therefore,
\begin{align}
A < \Bigl(1-\frac{1}{d}\Bigr)e^{-m\lambda_*} + \frac{1}{d}
= 1 - \Bigl(1-\frac{1}{d}\Bigr)\left(1 - e^{-m\lambda_*}\right).
\end{align}
Define:
\begin{align}
\label{eqn:alphaDef}
\alpha(m) := \Bigl(1-\frac{1}{d}\Bigr)\left(1 - e^{-m\lambda_*}\right),
\end{align}
so \(A < 1-\alpha(m)\) and hence (for any positive $L$):
\begin{align}
A^L < (1-\alpha(m))^L < e^{-L\alpha(m)}.
\end{align}
Thus, for any  \(\delta\in(0,1/2)\), \(A^L < \delta\) is guaranteed whenever:
\begin{equation}
L\alpha(m) > \ln\frac{1}{\delta}.
\label{eq:Lalpha}
\end{equation}
Substituting the definition of \(\alpha(m)\), in Eqn.~\eqref{eqn:alphaDef} into Eqn.~\eqref{eq:Lalpha} gives:
\begin{equation}
L\Bigl(1-\frac{1}{d}\Bigr)\left(1 - e^{-m\lambda_*}\right)
> \ln\frac{1}{\delta}.
\label{eq:Lalpha-expanded}
\end{equation}
For finite \(m\), we have \(1 - e^{-m\lambda_*} < 1\), thus, $L$ must satisfy:
\begin{align}
L > \frac{\ln(1/\delta)}{1-\frac{1}{d}}
\;\in\; \bigl[\ln(1/\delta),\,2\ln(1/\delta)\bigr].
\end{align}
Thus, a lower bound on number of states required to solve problem \ref{prob:P2.1}, \(L\), scales proportionally to \(\log(1/\delta)\). For any such feasible pair, \((m,L)\), define:
\begin{align}
x := \frac{\ln(1/\delta)}{L\left(1-\frac{1}{d}\right)}.
\end{align}
Then \cref{eq:Lalpha-expanded} implies:
\begin{align}
\label{eqn:mLambdaxRelation}
1 - e^{-m\lambda_*} > x.
\end{align}
For finite \(m\), we have \(1 - e^{-m\lambda_*} < 1\), hence \(0 < x < 1\). Rearranging the inequality in Eqn.~\eqref{eqn:mLambdaxRelation} we get:
\begin{align}
e^{-m\lambda_*} < 1 - x,
\qquad
m\lambda_* > -\ln(1-x)
= \ln\Bigl(\tfrac{1}{1-x}\Bigr).
\end{align}
Multiplying this by \(L\) and re-arranging yields:
\begin{align}
mL > \frac{L}{\lambda_*}\,\ln\Bigl(\tfrac{1}{1-x}\Bigr).
\end{align}
This implies, additionally using that $\forall x \in (0,1)$ \(-\ln(1-x)> x\), gives:
\begin{align}
mL > 
\frac{L}{\lambda_*}\,
\frac{\ln(1/\delta)}{L\left(1-\frac{1}{d}\right)}
= \frac{\ln(1/\delta)}{\lambda_*\left(1-\frac{1}{d}\right)}.
\end{align}
Finally, substituting \(\lambda_* = \dfrac{\varepsilon}{2\left(1-\frac{1}{d^2}\right)}\) into this inequality yields:
\begin{align}
    mL > \frac{2\bigl(1-\tfrac{1}{d^2}\bigr)}{\varepsilon\bigl(1-\tfrac{1}{d}\bigr)}\,
\ln\frac{1}{\delta}
= \frac{2(d+1)}{d\,\varepsilon}\,\ln\frac{1}{\delta}.
\end{align}
Combining the lower bound on $L$ with the product bound, we see that:
\begin{align}
L = \Omega\bigl(\ln\tfrac{1}{\delta}\bigr)
\quad\text{and, for such }L,\quad
m = \Theta\Bigl(\tfrac{1}{\varepsilon}\Bigr).
\end{align}
In particular, the minimum number of states $L$ required for the testing task
grows logarithmically in $1/\delta$, while the depth parameter $m$ (i.e. the number of black-box applications in a sequence) only needs
to scale on the order of $1/\varepsilon$. We have already shown in the previous
section that if $m=1$, the entire $1/\varepsilon$ scaling must be absorbed
by the number of states $L$, giving
\begin{align}
L = \Theta \Bigl(\tfrac{1}{\varepsilon}\ln\tfrac{1}{\delta}\Bigr).
\end{align}
This scaling is not tied to either extreme: it holds throughout the full tradeoff curve. Any choice of $(m,L)$ --either satisfying $m=\mathcal{O}(1/\varepsilon)$ and $L=\Theta(\ln(1/\delta))$, or $m=\mathcal{O}(1)$ and $L=\Theta((1/\varepsilon)\ln(1/\delta))$, or any intermediate allocation between these limits -- achieves the same overall scaling $mL=\Theta\bigl((1/\varepsilon)\ln(1/\delta)\bigr)$. 
\end{proof}

\subsection{Tolerant Tester}
\label{sec:tol_test_UB}
\begin{apptheorem}
    For the tolerant testing problem in \emph{\ref{prob:P2.2}}, we give an ancilla-free, non-adaptive
algorithm using individual measurements that achieves an upper bound of
$O\bigl(\varepsilon_2/(\varepsilon_2 - \varepsilon_1)^2\bigr)$ queries. 
\label{thm:tol_test}
\end{apptheorem}
\begin{proof}
    Let \(p\) be the probability of obtaining an outcome $1$, at each step, when the depolarizing parameter is $\lambda$, i.e., 
\begin{align}p=1-\langle 0 | \mathcal{N}_{\lambda}(|0\rangle\langle 0|) | 0\rangle = \lambda(1-1/d).\end{align}

The probability of error under $H_0$ is
\begin{align}
\text{Pr}^{\text{err}}_{H_0}
&= \Pr_p\!\left(\dfrac{S}{N} \ge \tau \,\Big|\, p \le \dfrac{\varepsilon_1}{2c}\right)=
\sup_{p \,\le\, \varepsilon_1/2c}
\Pr_p\!\left(\dfrac{S}{N} \ge \tau \right)\\&=
\Pr_{p \,=\, \varepsilon_1/2c}\!\left(\dfrac{S}{N} \ge \tau \right)
\label{Monotonicty_PrH0} \\
&\le \exp\!\big[-N\, D_{\mathrm{KL}}(\tau \,\Vert\, \varepsilon_1/2c)\big]
\;\le\; \delta,
\label{Chernoff-Heoffding_PrH0}
\end{align}
where c = {1 + 1/d}.  Eqn.~\eqref{Monotonicty_PrH0} follows from the monotonicity of the binomial upper tail in $p$ (see Lemma~\ref{lem:binomial_monot_upper}): the worst case over the null hypothesis region $\{p \le \varepsilon_1/2c\}$ is attained at the boundary $p = \varepsilon_1/2c$. Let $\tau > \tfrac{\varepsilon_1}{2c}$. Eqn.~\eqref{Chernoff-Heoffding_PrH0} is the Chernoff--Hoeffding bound for a binomial distribution (see Lemma~\ref{Chernoff-Hoeffding-Bound}). 

Similarly using Lemma~\ref{lem:binomial_monot_lower}, and the choice, $\tau < \tfrac{\varepsilon_2}{2c}$ to invoke Lemma~\ref{Chernoff-Hoeffding-Bound} we get the probability of error, under $H_1$, as:
\begin{align}
\text{Pr}^{\text{err}}_{H_1}
&=
\sup_{p \ge \varepsilon_2/2c}
\Pr_p\left(\frac{S}{N} < \tau \right) \\&\le \sup_{p \ge \varepsilon_2/2c}
\Pr_p\left(\frac{S}{N} \le \tau \right)
= 
\Pr_{p = \varepsilon_2/2c}\left(\frac{S}{N} \le \tau \right)
\\
&\le 
\exp\big[-N D_{KL}(\tau \Vert \varepsilon_2/2c)\big]
\le \delta. \label{Chernoff-Heoffding_PrH1}
\end{align}

Here $\tau \in (\varepsilon_1/2c, \varepsilon_2/2c)$. Therefore, we can now bound the query complexity of problem \ref{prob:P2.2} as:
\begin{align}
N
&=
\left\lceil
\max\left\{
\dfrac{\log(1/\delta)}
{D_{KL}\left(\tau\,\middle\|\,\frac{\varepsilon_1}{2c}\right)},
\dfrac{\log(1/\delta)}
{D_{KL}\left(\tau\,\middle\|\,\frac{\varepsilon_2}{2c}\right)}
\right\}
\right\rceil \\
&=
\left\lceil
\dfrac{\log(1/\delta)}
{\min\left\{
D_{KL}\left(\tau\,\middle\|\,\frac{\varepsilon_1}{2c}\right),
D_{KL}\left(\tau\,\middle\|\,\frac{\varepsilon_2}{2c}\right)
\right\}}
\right\rceil .
\end{align}
From \cref{Chernoff-Heoffding_PrH0} and \cref{Chernoff-Heoffding_PrH1} we can choose the value of $\tau$ that satisfies:
\begin{align}
D_{KL}\left(\tau \,\middle\|\, \tfrac{\varepsilon_1}{2c}\right) 
= D_{KL}\left(\tau \,\middle\|\, \tfrac{\varepsilon_2}{2c}\right)
\Rightarrow \tau = \frac{\ln\left(\frac{2c-\varepsilon_1}{2c-\varepsilon_2}\right)}
{\ln\left(\frac{\varepsilon_2(2c-\varepsilon_1)}{\varepsilon_1(2c-\varepsilon_2)}\right)}.
\end{align}
For the above choice of $\tau$, we have
\[
\min\left\{
D_{KL}\left(\tau\,\middle\|\,\frac{\varepsilon_1}{2c}\right),
D_{KL}\left(\tau\,\middle\|\,\frac{\varepsilon_2}{2c}\right)
\right\}
=
\Omega\left(
\frac{(\varepsilon_2-\varepsilon_1)^2}{c\varepsilon_2}
\right).\]
Therefore,
\begin{align}
N=
O\left(
\frac{\varepsilon_2\log(1/\delta)}
{(\varepsilon_2-\varepsilon_1)^2}
\right),
\end{align}
\end{proof}

\section{Lower Bounds}
\label{app:LowerBound}
\begin{apptheorem}
    For the tolerant testing problem in \emph{\ref{prob:P2.2}}, a matching lower
bound of $\Omega \bigl(\varepsilon_2/(\varepsilon_2 - \varepsilon_1)^2\bigr)$, up to
constant factors, holds even in the fully general ancilla-assisted, adaptive model with
arbitrary individual measurements.
\label{Th:LB_Tol}
\end{apptheorem}
\begin{proof}
The hardest channels to test are the ones that sit right on the boundaries of the two hypothesis. We denote the hard case channel constructions as $\mathcal{N}_{\lambda_1} \text{ and } \mathcal{N}_{\lambda_2}$, satisfying the relations:
\begin{align}
    &\operatorname{d}_{\diamond}(\mathcal{N}_{\lambda_1}, \operatorname{id}) = \varepsilon_1 \implies \lambda_1 = 1/2\varepsilon_1 (1-1/d^2 )^{-1} \label{lambda_1}\\& \operatorname{d}_{\diamond}(\mathcal{N}_{\lambda_2}, \operatorname{id}) = \varepsilon_2 \implies \lambda_2 = 1/2\varepsilon_2\left(1-{1}/{d^2}\right)^{-1}\label{lambda_2}
\end{align}
Where $\varepsilon_2 > \varepsilon_1$. A $\delta$-correct algorithm must distinguish between $\mathcal{N}_{\lambda_1}$ and $\mathcal{N}_{\lambda_2}$ with probability of success at least $1-\delta$. 

We work in the composite Hilbert space 
$\mathcal{H}' \otimes \mathcal{H} = \mathbb{C}^{d'} \otimes \mathbb{C}^{d}$,
where $d', d \in \mathbb{N}$ denote the ancilla and system dimensions, respectively. Here, \( d' \) may assume any arbitrarily large finite value.
At each step $t \in [N]$, we send in a state
\(\ket{\psi_t} \in \mathcal{H}' \otimes \mathcal{H}\)
which can be assumed to be pure by the convexity of KL divergence. Then we apply the channel $\operatorname{id}_{d'}\otimes \mathcal{N}_{\lambda}$,
and measure the output using the POVM:
\begin{align}
\mathbb{M}^{(t)}
= \left\{
M^{(t)}_{i_t}\right\}_{i_t}^{\kappa_t}
= \left\{c^{(t)}_{i_t}\,|\phi_{i_t}\rangle\!\langle \phi_{i_t}|
\right\}_{i_t = 1}^{\kappa_t},
\end{align}
Here, the superscript \((t)\) labels the measurement performed at step \(t\); it is not an exponent. The parameter
\(\kappa_t \in \mathbb{N}\) denotes the number of possible measurement outcomes at step \(t\), and \(i_t \in [\kappa_t]\) denotes the observed outcome. For each \(i_t\),
\[
\ket{\phi_{i_t}} \in \mathcal{H}' \otimes \mathcal{H},
\qquad
c^{(t)}_{i_t} \ge 0,
\]
and the POVM elements satisfy the completeness condition
\[
\sum_{i_t=1}^{\kappa_t} M^{(t)}_{i_t}
= \mathbb{I}_{\mathcal{H}' \otimes \mathcal{H}}.
\]

Denote the sequence of outcomes of the measurements by:
$\mathbf{X} = (X_1 = i_1, X_2 = i_2, \ldots, X_N = i_N)$,
where each $X_t$ is a $\{1,2,\ldots,\kappa_t\}$-valued random variable representing on of the measurements and $i_t$ is the corresponding measurement outcome.

\emph{Single-shot law: }For a fixed $\lambda \in [0,1]$, the probability of obtaining outcome $i_t$ (possibly conditioned on previous outcomes $X_{<t} = i_{<t}$) is:
\begin{align}
p^{(t)}_{\lambda}(i_t \mid X_{<t})
&= \operatorname{Tr}\Big[
M^{(t)}_{i_t}\,
(\operatorname{id}\otimes\mathcal{N}_{\lambda})(\ket{\psi_t}\!\bra{\psi_t})
\Big] \\
&= \lambda \,u^{(t)}_{i_t}(X_{<t})
    + (1-\lambda)\,a^{(t)}_{i_t}(X_{<t}),
\label{eq:single-shot}
\end{align}
where:
\begin{align}
u^{(t)}_{i_t}(X_{<t})
&:= c^{(t)}_{i_t}\,
\bra{\phi_{i_t}}(\operatorname{id}\otimes\tfrac{\mathbb{I}_d \operatorname{Tr(\cdot)}}{d})(\ket{\psi_t}\bra{\psi_t})\ket{\phi_{i_t}},\\
a^{(t)}_{i_t}(X_{<t})
&:= c^{(t)}_{i_t}\,|\langle \phi_{i_t}|\psi_t\rangle|^2,
\end{align}
and $\sum_{i_t} u^{(t)}_{i_t} = \sum_{i_t} a^{(t)}_{i_t}=1$.\\~\\
\emph{Lower Bound on KL Divergence: }
Define the joint distributions corresponding to the depolarizing parameters $\lambda_1$ and $\lambda_2$ as defined in \cref{lambda_1} and \cref{lambda_2}:
\begin{align}
P(\mathbf X)
:= \prod_{t=1}^N p^{(t)}_{\lambda_1}(i_t \mid X_{<t})~~\text{and}~~
Q(\mathbf X)
:= \prod_{t=1}^N p^{(t)}_{\lambda_2}(i_t \mid X_{<t}).
\end{align}
From the above definitions:
\begin{align}
D_{KL}(P\|Q)
&:= \sum_{\mathbf{X}}
P(\mathbf{X})\,
\log\frac{P(\mathbf{X})}{Q(\mathbf{X})}.
\label{eq:def-kl}
\end{align}
The Kullback–Leibler (KL) divergence is an $f$-divergence with generator $f(t)=t\log t$. By Remark 4.4 in \cite{wu2020itstats}, we obtain the following lower bound on the KL divergence using the data-processing inequality for $f$-divergences:
\begin{align}
    D_{KL}(P||Q) &\geq D_{KL}(\operatorname{P(E)}||\operatorname{Q(E)}) \\& \ge D_{KL}(\operatorname{Bern}(1- \delta) || \operatorname{Bern}(\delta)).
    \label{Dkl,lower}
\end{align}
We henceforth denote $D_{KL}(\operatorname{Bern}(1- \delta) || \operatorname{Bern}(\delta))$ by $g(\delta)$. The next steps are now to upper bound the KL Divergence as a function of $(N,\varepsilon)$.

\emph{Kl Divergence under adaptivity:}
The logarithmic term in \cref{eq:def-kl} can be written as a sum by pulling the product outside the log:
\begin{align}
\log\frac{P(\mathbf{X})}{Q(\mathbf{X})}
&= \log\left(\prod_{t=1}^N
\frac{p^{(t)}_{\lambda_1}(i_t\mid i_{<t})}
{p^{(t)}_{\lambda_2}(i_t\mid i_{<t})}\right)
\\&= \sum_{t=1}^N
\log\frac{p^{(t)}_{\lambda_1}(i_t\mid i_{<t})}
{p^{(t)}_{\lambda_2}(i_t\mid i_{<t})}.
\label{eq:log-expand}
\end{align}
Substituting \cref{eq:log-expand} into \cref{eq:def-kl} gives:
\begin{align}
D_{KL}(P\|Q)
&= \sum_{\mathbf{X}} P(\mathbf{X})
\sum_{t=1}^N
\log\frac{p^{(t)}_{\lambda_1}(i_t\mid i_{<t})}
{p^{(t)}_{\lambda_2}(i_t\mid i_{<t})}
\\&= \sum_{t=1}^N
\sum_{\mathbf{X}} P(\mathbf{X})
\log\frac{p^{(t)}_{\lambda_1}(i_t\mid i_{<t})}
{p^{(t)}_{\lambda_2}(i_t\mid i_{<t})}.
\label{eq:swap-sum}
\end{align}

Now we expand the sum over all possible outcome strings
$\mathbf{X}=(i_1,\ldots,i_N)$.
For any function $f(i_{<t},i_t)$ depending only on the history $i_{<t}$
and the current outcome $i_t$, we have:
\begin{align}
&\sum_{\mathbf{X}} P(\mathbf{X}) f(i_{<t},i_t) \\ &= 
\sum_{i_1}\sum_{i_2}\cdots\sum_{i_N}
\Bigg(
\prod_{s=1}^N p^{(s)}_{\lambda_1}(i_s\mid i_{<s})
\Bigg)
f(i_{<t},i_t).
\end{align}
We separate this sum into three parts:
past $(i_{<t})$, current $(i_t)$, and future $(i_{>t})$ outcomes:
\begin{equation}
\begin{aligned}
\sum_{\mathbf{X}} P(\mathbf{X})\, f(i_{<t},i_t)
&= \sum_{i_{<t}}\sum_{i_t}\sum_{i_{>t}}
\Bigg(\prod_{s\le t} p^{(s)}_{\lambda_1}(i_s\mid i_{<s}) \Bigg) \\
&\quad\cdot
\Bigg(\prod_{s> t} p^{(s)}_{\lambda_1}(i_s\mid i_{<s}) \Bigg)\,
f(i_{<t},i_t).
\end{aligned}
\end{equation}
The inner sum over future outcomes $i_{>t}$ equals $1$, since conditional probabilities normalize:
\begin{align}
\sum_{i_{>t}} \prod_{s>t} p^{(s)}_{\lambda_1}(i_s\mid i_{<s}) = 1.
\end{align}
Hence,
\begin{align}
&\sum_{\mathbf{X}} P(\mathbf{X}) f(i_{<t},i_t)
\\&= \sum_{i_{<t}}
P_{<t}(i_{<t})
\sum_{i_t} p^{(t)}_{\lambda_1}(i_t\mid i_{<t}) f(i_{<t},i_t),
\label{eq:marginalization}
\end{align}
where $P_{<t}(i_{<t}) = \prod_{s<t} p^{(s)}_{\lambda_1}(i_s\mid i_{<s})$ is the marginal
distribution over the first $t-1$ outcomes. 

Applying \cref{eq:marginalization} to the function:
\begin{align}
f(i_{<t},i_t)
= \log\frac{p^{(t)}_{\lambda_1}(i_t\mid i_{<t})}
{p^{(t)}_{\lambda_2}(i_t\mid i_{<t})},
\end{align}
and substituting back into \cref{eq:swap-sum}, we obtain:
\begin{align}
&D_{KL}(P\|Q)
\\&= \sum_{t=1}^N
\sum_{i_{<t}} P_{<t}(i_{<t})
\sum_{i_t} p^{(t)}_{\lambda_1}(i_t\mid i_{<t})
\log\frac{p^{(t)}_{\lambda_1}(i_t\mid i_{<t})}
{p^{(t)}_{\lambda_2}(i_t\mid i_{<t})}.
\label{eq:marginalized}
\end{align}
Recognizing the inner sum in Eqn.~\eqref{eq:marginalized} as the conditional KL divergence completes the derivation of the chain rule:
\begin{equation}
\begin{aligned}
&D_{KL}(P\|Q)
\\&= \sum_{t=1}^N
\mathbb{E}_{X_{<t}\sim P}\left[
D_{KL}\Big(
p^{(t)}_{\lambda_1}(\cdot\mid X_{<t})
\,\Big\|\,
p^{(t)}_{\lambda_2}(\cdot\mid X_{<t})
\Big)
\right].
\label{eq:chain-rule}
\end{aligned}
\end{equation}

\emph{Upper Bounding each conditional KL term.}

Next is to bound each term $D_{KL}\Big(
p^{(t)}_{\lambda_1}(\cdot\mid X_{<t})
\,\Big\|\,
p^{(t)}_{\lambda_2}(\cdot\mid X_{<t})
\Big)$. Fix a step $t\in [N]$.  
From the single--shot law \cref{eq:single-shot}, for each $\lambda_k\in[0,1]$ redefine:
\begin{align}
p_k(i_t)
:&= p^{(t)}_{\lambda_k}(i_t \mid X_{<t})
\\&= \lambda_k\,u^{(t)}_{i_t}(X_{<t})
  + (1-\lambda_k)\,a^{(t)}_{i_t}(X_{<t}),
\end{align}
where $k\in\{1,2\}$ and:
\begin{align}
u^{(t)}_{i_t}(X_{<t})
& = c^{t}_{i_t}\,
\langle \phi_{i_t}|\,
(\operatorname{id}\otimes \tfrac{\mathbb{I}_d \operatorname{Tr(\cdot)}}{d})(|\psi_t\rangle\!\langle \psi_t|)
\,|\phi_{i_t}\rangle\\
\qquad
a^{(t)}_{i_t}(X_{<t})
&= c^{t}_{i_t}\,|\langle \phi_{i_t}\mid \psi_t\rangle|^2\\
\text{and } \sum_{i_t} u^{(t)}_{i_t}& =\sum_{i_t} a^{(t)}_{i_t}=1.
\end{align}
The conditional KL divergence from Eqn. \eqref{eq:chain-rule} at step $t$ is:
\begin{align}
D_{KL}\Big(
p^{(t)}_{\lambda_1}(\cdot\mid X_{<t})
\,\Big\|\,
p^{(t)}_{\lambda_2}(\cdot\mid X_{<t})
\Big)
= \sum_{i_t} p_1(i_t)\log\frac{p_1(i_t)}{p_2(i_t)},
\label{p_k}
\end{align}
where each distribution, $p_k$, is a convex combination of $u^{(t)}$ and $a^{(t)}$ i.e., $\forall k\in\{1,2\}, t \in [N]$, $p_k(i_t)$ may be expressed as:
\begin{align}
p_k(i_t) = \lambda_k\,u^{(t)}_{i_t} + (1-\lambda_k)\,a^{(t)}_{i_t}.
\end{align}

To enable us to be concise and more readable, we first pause to define the below notation:
\begin{align}
\alpha_1(i_t) := \lambda_1\,u^{(t)}_{i_t}, 
\qquad 
\beta_1(i_t) := (1-\lambda_1)\,a^{(t)}_{i_t},
\end{align}
\begin{align}
\alpha_2(i_t) := \lambda_2\,u^{(t)}_{i_t}, 
\qquad 
\beta_2(i_t) := (1-\lambda_2)\,a^{(t)}_{i_t}.
\end{align}
Therefore, $p_1(i_t)$ and $p_2(i_t)$ can be expressed more conveniently as:
\begin{align}
p_1(i_t)=\alpha_1(i_t)+\beta_1(i_t),
\qquad
p_2(i_t)=\alpha_2(i_t)+\beta_2(i_t).
\end{align}
Hence the KL divergence in \cref{p_k} can be expressed as:
\begin{align}
\label{eqn:KLDivergenceAltExpress}
\sum_{i_t}
\Big(\alpha_1(i_t) + \beta_1(i_t)\Big)
\log\frac{\alpha_1(i_t)+\beta_1(i_t)}{\alpha_2(i_t)+\beta_2(i_t)}.
\end{align}

Applying the log--sum inequality to \cref{p_k}, and using \cref{eqn:KLDivergenceAltExpress} we obtain:
\begin{align}
\sum_{i_t} p_1(i_t)\log\frac{p_1(i_t)}{p_2(i_t)}
&\le 
\sum_{i_t} \alpha_1(i_t)\log\frac{\alpha_1(i_t)}{\alpha_2(i_t)} \nonumber\\
&+
\sum_{i_t} \beta_1(i_t)\log\frac{\beta_1(i_t)}{\beta_2(i_t)}.
\label{eq:logsum-expanded}
\end{align}
To progress, we evaluate each term in \cref{eq:logsum-expanded} individually:
\begin{align}
    \sum_{i_t} \alpha_1(i_t)\log\frac{\alpha_1(i_t)}{\alpha_2(i_t)}
&=
\sum_{i_t} \lambda_1u^{(t)}_{i_t}
\log \frac{\lambda_1 u^{(t)}_{i_t}}{\lambda_2u^{(t)}_{i_t}}
\\&= \lambda_1\log\frac{\lambda_1}{\lambda_2},
\end{align}
since $\sum_{i_t}u^{(t)}_{i_t}=1$.  
Similarly,
\begin{align}
\sum_{i_t} \beta_1(i_t)\log\frac{\beta_1(i_t)}{\beta_2(i_t)}
&=
\sum_{i_t} (1-\lambda_1) a^{(t)}_{i_t}
\log\frac{(1-\lambda_1) a^{(t)}_{i_t}}{(1-\lambda_2) a^{(t)}_{i_t}}
\\&= (1-\lambda_1)\log\frac{(1-\lambda_1)}{(1-\lambda_2)}.
\end{align}
Substituting this back into \cref{eq:logsum-expanded} yields the below upper bound of the conditional KL divergence at step $t$:
\begin{align}
\sum_{i_t} p_1(i_t)\log\frac{p_1(i_t)}{p_2(i_t)}
&\le (1-\lambda_1)\log\frac{1-\lambda_1}{1-\lambda_2}
   + \lambda_1\log\frac{\lambda_1}{\lambda_2} \\
&= D_{KL}\big(\mathrm{Bern}(\lambda_1)\,\|\,\mathrm{Bern}(\lambda_2)\big).
\label{eq:bernoulli-bound}
\end{align}

Thus, every conditional KL term is upper bounded by the KL divergence
between the Bernoulli distributions with parameters $\lambda_1$ and $\lambda_2$.

\emph{Combining the bounds.}
Substituting \cref{eq:bernoulli-bound} into \cref{eq:chain-rule} gives:
\begin{align}
\label{eqn:KLBoundSumExpect}
D_{KL}(P\|Q)
&\le
\sum_{t=1}^N
\mathbb{E}_{X_{<t}\sim P}
\Big[
D_{KL}\big(\mathrm{Bern}(\lambda_1)\,\|\,\mathrm{Bern}(\lambda_2)\big)
\Big].
\end{align}
Since the KL divergence between two Bernoulli distributions is independent of $X_{<t}$:
\begin{align}
D_{KL}(P\|Q)
&\le
\sum_{t=1}^N
D_{KL}\big(\mathrm{Bern}(\lambda_1)\,\|\,\mathrm{Bern}(\lambda_2)\big)
\\&=
N\,
D_{KL}\big(\mathrm{Bern}(\lambda_1)\,\|\,\mathrm{Bern}(\lambda_2)\big).
\label{eq:kl-upper}
\end{align}

\emph{Relating to the tolerance parameter.}
We now use the fact that the Kullback-Leibler divergence is bounded above by the $\chi^2$-divergence, $D_{KL}(P\|Q) \le \chi^2(P \| Q)$, which for Bernoulli distributions yields:
\begin{equation}
D_{KL}\big(\text{Bern}(\lambda_1)\,\|\,\text{Bern}(\lambda_2)\big) \le \frac{(\lambda_2-\lambda_1)^2}{\lambda_2(1-\lambda_2)}.
\end{equation}
Substituting the respective $\lambda$'s relationship with $\varepsilon$'s into= \cref{lambda_1} and \cref{lambda_2} we get, for $d\geq2$:
\begin{align}
g(\delta)
\le
D_{KL}(P\|Q)
<
N\,\frac{(\varepsilon_2-\varepsilon_1)^2}{\varepsilon_2(1.5-\varepsilon_2)}
\label{eq:Dkl-upper}
\end{align}
which implies the query complexity lower bound:
\begin{align}
N
=
\Omega\left(
\dfrac{\varepsilon_2}{(\varepsilon_2 - \varepsilon_1)^2}\,
g(\delta)
\right).
\label{eq:n-lower}
\end{align}
Here the strict inequality is due to the fact that $\varepsilon_1$ is strictly less than $\varepsilon_2$.
\end{proof}

\begin{apptheorem}
\label{Th:LB_Std}
    For the problem in \emph{\ref{prob:P2.1}}, a lower bound of $\Omega(1/\varepsilon)$, up to constant factors, holds even in the fully general ancilla-assisted, adaptive model with arbitrary individual measurements. 
\end{apptheorem}
\begin{proof}
    Setting \(\varepsilon_2 = \varepsilon\) and \(\varepsilon_1 = 0\) in Theorem~\ref{Th:LB_Tol} yields the corresponding lower bound for the standard tester.
\end{proof}
\end{document}